\documentclass[10pt]{amsart}
\usepackage{epsf,latexsym,amsmath,amssymb,amscd,datetime}
\usepackage{amsmath,amsthm,amssymb,enumerate,eucal,url,calligra,mathrsfs}

\usepackage{blkarray} 

\usepackage{tikz,scalerel}
\usetikzlibrary{arrows}

\usepackage{imakeidx}     

\usepackage[outdir=./]{epstopdf}

\usepackage{graphicx}
\usepackage{color}
\newenvironment{jfnote}{ \bgroup \color{blue} }{\egroup}

\newcommand{\darkgreen}{\color[rgb]{0.0,0.45,0.0}} 

\newcommand{\vis}{{\rm visits}}

\newcommand{\red}{\color[rgb]{1.0,0.2,0.2}} 
\newcommand{\blue}{\color[rgb]{0.2,0.2,1.0}} 

\newcommand{\oldStuff}[1]{}

\newcommand{\Coord}{{\rm Coord}}
 
\newcommand{\og}{{\scriptscriptstyle \le}}
\newcommand{\Bg}{{\scalebox{1.0}{$\!\scriptscriptstyle /\!B$}}}

\newcommand{\cert}{\xi}
 
\DeclareMathOperator{\SHom}{\mathscr{H}\text{\kern -3pt {\calligra\large om}}\,}

\IfFileExists{my_xrefs}{\input my_xrefs}{}

\newcommand{\subgr}{{\rm subgr}}

\DeclareMathOperator{\ViSu}{VisSub}

\DeclareMathOperator{\legal}{legal}

\newcommand{\naturals}{{\mathbb N}}

\newcommand{\Eor}{E^{\mathrm{or}}}
\newcommand{\mec}[1]{{\bf #1}}	
\newcommand{\bec}[1]{{\boldsymbol #1}}	

\DeclareMathOperator{\Trace}{Trace}

\usepackage{mathrsfs}
\usepackage{amssymb}
\usepackage{dsfont}
\usepackage{verbatim}
\usepackage{url}

\newcommand{\Edir}{E^{\mathrm{dir}}}

\theoremstyle{plain}
\newtheorem{theorem}{Theorem}[section]
\newtheorem{lemma}[theorem]{Lemma}

\newtheorem{corollary}[theorem]{Corollary}

\theoremstyle{definition}
\newtheorem{definition}[theorem]{Definition}

\newtheorem{xca}{Exercise}[section]

\newtheorem{example}[theorem]{Example}

\newcommand{\isom}{\simeq} 

\newcommand{\ignore}[1]{}

\newcommand{\reals}{{\mathbb R}}

\newcommand{\integers}{{\mathbb Z}}

\newcommand{\complex}{{\mathbb C}}

\newcommand\EE{\mathbb{E}}

\newcommand\II{\mathbb{I}}

\DeclareMathAlphabet{\mathcal}{OMS}{cmsy}{m}{n}

\newcommand\cA{\mathcal{A}}

\newcommand\cC{\mathcal{C}}
\newcommand\cD{\mathcal{D}}

\newcommand\cN{\mathcal{N}}

\newcommand\cP{\mathcal{P}}

\newcommand\cR{\mathcal{R}}
\newcommand\cS{\mathcal{S}}
\newcommand\cT{\mathcal{T}}
\newcommand\cU{\mathcal{U}}

\DeclareMathOperator{\Prob}{Prob}

\DeclareMathOperator{\VLG}{VLG}

\DeclareMathOperator{\Line}{Line}

\DeclareMathOperator{\SNBC}{SNBC}

\DeclareMathOperator{\snbc}{snbc}

\def\from{\colon}

\def\isom{\simeq}

\def\eqdef{\overset{\text{def}}{=}}

\DeclareMathOperator{\ord}{ord}

\DeclareRobustCommand
  \rddots{\mathinner{\mkern1mu\raise\p@
    \vbox{\kern7\p@\hbox{.}}\mkern2mu
    \raise4\p@\hbox{.}\mkern2mu\raise7\p@\hbox{.}\mkern1mu}}

\newcommand\xhookrightarrow[2][]{\ext@arrow 0062{\hookrightarrowfill@}{#1}{#2}}
\def\hookrightarrowfill@{\arrowfill@\lhook\relbar\rightarrow}

\usepackage{relsize}
\usepackage{tikz}
\usetikzlibrary{matrix,arrows,decorations.pathmorphing}
\usepackage{tikz-cd}
\usetikzlibrary{cd}

\usepackage[pdftex,colorlinks,linkcolor=blue]{hyperref}

\newcommand{\myDeleteNote}[1]{{\red #1}}

\begin{document}

\title[Relativized Alon Conjecture II] 
{On the Relativized Alon Eigenvalue
Conjecture II: Asymptotic Expansion Theorems for Walks}

\author{Joel Friedman}
\address{Department of Computer Science,
        University of British Columbia, Vancouver, BC\ \ V6T 1Z4, CANADA}
\curraddr{}
\email{{\tt jf@cs.ubc.ca}}
\thanks{Research supported in part by an NSERC grant.}

\author{David Kohler}
\address{Department of Mathematics,
        University of British Columbia, Vancouver, BC\ \ V6T 1Z2, CANADA}
\curraddr{422 Richards St, Suite 170, Vancouver BC\ \  V6B 2Z4, CANADA}
\email{{David.kohler@a3.epfl.ch}}
\thanks{Research supported in part by an NSERC grant.}

%
\date{\today}



\begin{abstract}

This is the second in a series of articles devoted to showing that a typical
covering map of large degree to a fixed, regular graph has its new adjacency
eigenvalues within the bound conjectured by Alon for random regular graphs.

The first main result in this article concerns the function $f(k,n)$ defined as
the number of SNBC 
(strictly non-backtracking closed) walks of length $k$ of a given
homotopy type in a random covering graph of degree $n$ of a fixed graph.  We
prove the existence of asymptotic expansions in powers of $1/n$ for $f(k,n)$,
where the coefficients---functions of $k$---are proven to have some desirable
properties; namely, these coefficients are approximately a sum of polynomials
times exponential functions.

The second main result is a generalization of the first, where the
number of SNBC walks of length $k$ is multiplied by an indicator
function that the covering graph contains a certain type 
of {\em tangle}; the second result requires more terminology, although
its proof uses the same basic tools used to prove the first result.

The results in this article are mostly straightforward generalizations of 
methods used in previous works.
However, this article
(1) ``factors'' these methods into a number of short, conceptually
simple, and independent parts,
(2) writes each independent part in more general terms, and
(3) significantly simplifies of one of the previous
computations.
As such we expect that this article will make it easier to apply 
trace methods to
related models of random graphs.

\end{abstract}

\maketitle
\setcounter{tocdepth}{3}
\tableofcontents

\newcommand{\sePrelimProofs}{17}

\section{Introduction}

This paper is the second in a series of six articles whose main
results are to prove a relativized version of Alon's Second
Eigenvalue Conjecture,
conjectured in \cite{friedman_relative}, in the case where
the base graph is regular.

For a detailed introduction to this series of articles, we
refer the reader to this first article in this series.

In this article we prove theorems regarding the existence
asymptotic expansions that count the expected number of walks
subject to certain conditions,
in families of random graphs.
These expansion theorems are the basis for proving the
existence of asymptotic expansion regarding {\em certified traces}
studied in the third article of this series.


The proofs of the main theorems of this article are based on
on the methods of \cite{friedman_random_graphs}.
However, in this article we 
``factor'' the methods
in \cite{friedman_random_graphs} into a number of independent parts, and
we have stated the results in each part in as general terms as seems
reasonable.
Furthermore, some of the computations in \cite{friedman_random_graphs}
have been simplified.
In this sense, the results of this article involve some new ideas
beyond a straightforward generalization of the methods of
\cite{friedman_random_graphs}.
Of course, this article is longer due to the fact that we need to
introduce the more general setting of random covering maps that
replace random graphs, and that our second main theorem counts,
roughly speaking,
the number of walks times the number of times the isomorphism class
of a fixed graph lies in the random graph.

We wish to stress that all the main techniques in this article
appear in
our proof of the first main theorem, which is
a simple generalization
of the main expansion theorem, Theorem~2.18, of
\cite{friedman_random_graphs}.
These techniques can be understood already just in the context
of random graphs as in \cite{friedman_random_graphs}, and the
reader may wish to keep this case in mind.

In this article we consider functions $f=f(k,n)$ such as the
expected number of strictly non-backtracking, closed walks of length
$k$ of a fixed homotopy type in a random graph admitting a covering
map of degree $n$ to a fixed graph, $B$; in this article we do not
require $B$ to be regular.  We show that such functions have
expansions
$$
c_0(k) + c_1(k)/n + \cdots + c_{r-1}(k)/n^{r-1} + O(1)c_r(k)/n^r
$$
to any order $r$, where the coefficients $c_i=c_i(k)$ have 
desirable properties.  It is easy to see that such expansions
exist; the main work is to prove the theorems we require regarding the
$c_i(k)$.  

We refer the reader to the first article in this series for
a detailed description of the definitions we use and their
motivation.
These definitions are reviewed in Section~\ref{se_defs_review}.
The rest of this article is organized as follows.
In Section~\ref{se_main_thms} we state the main theorems in this
article.
In Section~\ref{se_length_mult} we give the ``length-multiplicity
formula'' that is the foundation of our proofs of both main theorems.
In Section~\ref{se_proof_outline} we outline the
proof of our first main theorem,
which we give in
Sections~\ref{se_cert_dot}--\ref{se_main_cert_walks};
Section~\ref{se_cert_dot} is devoted to proving our
``Certified Dot Convolution'' Lemma, an abstract lemma about
a type of convolution of two multivariate functions,
and Section~\ref{se_regular} is a lemma about certain statistics
regarding regular languages; both of these lemmas serve to
factor the proof of the expansion theorem, Theorem~2.18,
of \cite{friedman_random_graphs}.
In Section~\ref{se_pairs_prelim} we give some preliminary
terminology and ideas needed to prove our second main 
theorem, which we complete in
Section~\ref{se_main_cert_pairs}.

\section{Review of the Main Definitions}
\label{se_defs_review}

We refer the reader to Article~I for the definitions used in this article,
the motivation of such definitions, and an appendix there that lists all the
definitions and notation.
In this section we briefly review these definitions and notation. 

\subsection{Basic Notation and Conventions}
\label{su_very_basic}

We use $\reals,\complex,\integers,\naturals$
to denote, respectively, the
the real numbers, the complex numbers, the integers, and positive
integers or
natural numbers;
we use $\integers_{\ge 0}$ ($\reals_{>0}$, etc.)
to denote the set of non-negative
integers (of positive real numbers, etc.).
We denote $\{1,\ldots,n\}$ by $[n]$.

If $A$ is a set, we use $\naturals^A$ to denote the set of
maps $A \to \naturals$; we will refers to its elements as
{\em vectors}, denoted in bold face letters, e.g., $\mec k\in \naturals^A$
or $\mec k\from A\to\naturals$; we denote its {\em component}
in the regular face equivalents, i.e., for $a\in A$,
we use $k(a)\in\naturals$ to denote
the $a$-component of $\mec k$.
As usual, $\naturals^n$ denotes $\naturals^{[n]}=\naturals^{\{1,\ldots,n\}}$.
We use similar conventions for $\naturals$ replaced by $\reals$,
$\complex$, etc.

If $A$ is a set, then $\# A$ denotes the cardinality of $A$.
We often denote a set with all capital letters, and its cardinality
in lower case letters; for example,
when we define
$\SNBC(G,k)$, we will write
$\snbc(G,k)$ for $\#\SNBC(G,k)$.

If $A'\subset A$ are sets, then $\II_{A'}\from A\to\{0,1\}$ (with $A$
understood) denotes
the characteristic function of $A'$, i.e., $\II_{A'}(a)$ is $1$ if
$a\in A'$ and otherwise is $0$;
we also write $\II_{A'}$ (with $A$ understood) to mean $\II_{A'\cap A}$
when $A'$ is not necessarily a subset of $A$.

All probability spaces are finite; hence a probability space
is a pair $\cP=(\Omega,P)$ where $\Omega$ is a finite set and
$P\from \Omega\to\reals_{>0}$ with $\sum_{\omega\in\Omega}P(\omega)=1$;
hence an {\em event} is any subset of $\Omega$.
We emphasize that $\omega\in\Omega$ implies that $P(\omega)>0$ with
strict inequality; we refer to the elements of $\Omega$ as
the atoms of the probability space.
We use $\cP$ and $\Omega$ interchangeably when $P$ is
understood and confusion is unlikely.

A {\em complex-valued random variable} on $\cP$ or $\Omega$
is a function $f\from\Omega\to\complex$, and similarly for real-,
integer-, and natural-valued random variable; we denote its
$\cP$-expected value by
$$
\EE_{\omega\in\Omega}[f(\omega)]=\sum_{\omega\in\Omega}f(\omega)P(\omega).
$$
If $\Omega'\subset\Omega$ we denote the probability of $\Omega'$ by
$$
\Prob_{\cP}[\Omega']=\sum_{\omega\in\Omega'}P(\omega')
=
\EE_{\omega\in\Omega}[\II_{\Omega'}(\omega)].
$$
At times we write $\Prob_{\cP}[\Omega']$ where $\Omega'$ is
not a subset of $\Omega$, by which we mean
$\Prob_{\cP}[\Omega'\cap\Omega]$.

\subsection{Graphs, Our Basic Models, Walks}

A {\em directed graph},
or simply a {\em digraph},
is a tuple $G=(V_G,\Edir_G,h_G,t_G)$ consisting of sets
$V_G$ and $\Edir_G$ (of {\em vertices} and {\em directed edges}) and maps
$h_G,t_G$ ({\em heads}
and {\em tails}) $\Edir_G\to V_G$.
Therefore our digraphs can have multiple edges and
self-loops (i.e., $e\in\Edir_G$ with $h_G(e)=t_G(e)$).
A {\em graph} is a tuple $G=(V_G,\Edir_G,h_G,t_G,\iota_G)$
where $(V_G,\Edir_G,h_G,t_G)$ is a digraph and
$\iota_G\from \Edir_G\to \Edir_G$ is an involution with
$t_G\iota_G=h_G$;
the {\em edge set} of $G$, denoted $E_G$, is the
set of orbits of $\iota_G$, which (notation aside)
can be identified with $\Edir_G/\iota_G$,
the set of equivalence classes of
$\Edir_G$ modulo $\iota_G$;
if $\{e\}\in E_G$ is a singleton, then necessarily $e$ is a self-loop
with $\iota_G e =e $, and
we call $e$ a {\em half-loop}; other elements of $E_G$ are sets
$\{e,\iota_G e\}$ of size two, i.e., with $e\ne\iota_G e$, and for such $e$
we say that $e$ (or, at times, $\{e,\iota_G e\}$)
is a {\em whole-loop} if
$h_G e=t_G e$ (otherwise $e$ has distinct endpoints).

Hence these definitions allow our graphs to have multiple edges and 
two types of self-loops---whole-loops
and half-loops---as in
\cite{friedman_geometric_aspects,friedman_alon}.
The {\em indegree} and {\em outdegree} of a vertex in a digraph is
the number of edges whose tail, respectively whose head, is the vertex;
the {\em degree} of a vertex in a graph is its indegree (which equals
its outdegree) in the underlying digraph; 
therefore a whole-loop about a vertex contributes $2$
to its degree, whereas a half-loop contributes $1$.

An {\em orientation} of a graph, $G$, is a choice $\Eor_G\subset\Edir_G$
of $\iota_G$ representatives; i.e., $\Eor_G$ contains every half-loop, $e$,
and one element of each two-element set $\{e,\iota_G e\}$.

A {\em morphism $\pi\from G\to H$} of directed graphs is a pair
$\pi=(\pi_V,\pi_E)$ where $\pi_V\from V_G\to V_H$ and
$\pi_E\from \Edir_G\to\Edir_H$ are maps that intertwine the heads maps
and the tails maps of $G,H$ in the evident fashion;
such a morphism is {\em covering} (respectively, {\em \'etale},
elsewhere called an {\em immersion}) if for each $v\in V_G$,
$\pi_E$ maps those directed edges whose head is $v$ bijectively
(respectively, injectively) to those whose head is $\pi_V(v)$,
and the same with tail replacing head.
If $G,H$ are graphs, then a morphism $\pi\from G\to H$ is a morphism
of underlying directed graphs where $\pi_E\iota_G=\iota_H\pi_E$;
$\pi$ is called {\em covering} or {\em \'etale} if it is so as a morphism
of underlying directed graphs.
We use the words {\em morphism} and {\em map} interchangeably.

A walk in a graph or digraph, $G$, is an alternating sequence
$w=(v_0,e_1,\ldots,e_k,v_k)$ of vertices and directed edges
with $t_Ge_i=v_{i-1}$ and $h_Ge_i=v_i$ for $i\in[k]$;
$w$ is {\em closed} if $v_k=v_0$;
if $G$ is a graph,
$w$ is {\em non-backtracking}, or simply {\em NB},
if $\iota_Ge_i\ne e_{i+1}$
for $i\in[k-1]$, and {\em strictly 
non-backtracking closed}, or simply {\em SNBC},
if it is closed, non-backtracking, and 
$\iota_G e_k\ne e_1$.
The {\em visited subgraph} of a walk, $w$, in a graph $G$, denoted
$\ViSu_G(w)$ or simply
$\ViSu(w)$, is the smallest subgraph of $G$ containing all the vertices
and directed edges of $w$;
$\ViSu_G(w)$ generally depends on $G$, i.e., $\ViSu_G(w)$ cannot be inferred
from the sequence $v_0,e_1,\ldots,e_k,v_k$ alone without knowing
$\iota_G$.

The adjacency matrix, $A_G$,
of a graph or digraph, $G$, is defined as usual (its $(v_1,v_2)$-entry
is the number of directed edges from $v_1$ to $v_2$);
if $G$ is a graph on $n$ vertices, 
then $A_G$ is symmetric and we order its eigenvalues (counted with
multiplicities) and denote them
$$
\lambda_1(G)\ge \cdots \ge \lambda_n(G).
$$
If $G$ is a graph, its
Hashimoto matrix (also called the non-backtracking matrix), $H_G$,
is the adjacency matrix of the {\em oriented line graph} of $G$,
$\Line(G)$,
whose vertices are $\Edir_G$ and whose directed edges
are the subset of $\Edir_G\times\Edir_G$ consisting of pairs $(e_1,e_2)$
such that $e_1,e_2$ form the
directed edges of a non-backtracking walk (of length two) in $G$
(the tail of $(e_1,e_2)$ is $e_1$, and its head $e_2$);
therefore $H_G$
is the square matrix indexed on $\Edir_G$, whose $(e_1,e_2)$ entry
is $1$ or $0$ according to, respectively, whether or not
$e_1,e_2$ form a non-backtracking walk
(i.e., $h_G e_1=t_G e_2$ and $\iota_G e_1\ne e_2$).
We use $\mu_1(G)$ to denote the Perron-Frobenius eigenvalue of 
$H_G$, and use $\mu_i(G)$ with $1<i\le \#\Edir_G$ to denote the
other eigenvalues of $H_G$ (which are generally complex-valued)
in any order.

If $B,G$ are both digraphs,
we say that $G$ is a {\em coordinatized graph over $B$
of degree $n$}
if
\begin{equation}\label{eq_coord_def}
V_G=V_B\times [n], \quad\Edir_G=\Edir_B\times[n], \quad
t_G(e,i)=(t_B e,i),\quad
h_G(e,i)=(h_Be,\sigma(e)i)
\end{equation} 
for some map
$\sigma\from\Edir_B\to\cS_n$, where $\cS_n$ is the group
of permutations on $[n]$; we call $\sigma$ (which is uniquely determined by
\eqref{eq_coord_def}) {\em the permutation assignment
associated to $G$}.
[Any such $G$ comes with a map $G\to B$ given by 
``projection to the first component of
the pair,'' and this map is a covering map of degree $n$.]
If $B,G$ are graphs, we say that a graph $G$ is a 
{\em coordinatized graph over $B$
of degree $n$} if \eqref{eq_coord_def} holds and also
\begin{equation}\label{eq_coord_def_graph}
\iota_G(e,i) = \bigl( \iota_B e,\sigma(e)i \bigr) ,
\end{equation} 
which implies that 
\begin{equation}\label{eq_sigma_iota_B}
(e,i)=\iota_G\iota_G(e,i) = \bigl( e, \sigma(\iota_B e)\sigma(e)i \bigr)
\quad\forall e\in\Edir_B,\ i\in[n],
\end{equation}
and hence $\sigma(\iota_B e)=\sigma(e)^{-1}$;
we use $\Coord_n(B)$ to denote the set of all coordinatized covers
of a graph, $B$, of degree $n$.

The {\em order} of a graph, $G$, is $\ord(G)\eqdef (\#E_G)-(\#V_G)$.
Note that a half-loop and a whole-loop each contribute $1$ to 
$\#E_G$ and to the order of $G$.
The {\em Euler characteristic} of a graph, $G$, is
$\chi(G)\eqdef (\# V_G) - (\#\Edir_G)/2$.
Hence $\ord(G)\ge -\chi(G)$, with equality iff $G$ has no half-loops.

If $w$ is a walk in any $G\in\Coord_n(B)$, then one easily
sees that $\ViSu_G(w)$ can be inferred
from $B$ and $w$ alone.

If $B$ is a graph without half-loops, then the {\em permutation model over
$B$} refers to the probability spaces $\{\cC_n(B)\}_{n\in\naturals}$ where
the atoms of $\cC_n(B)$ are coordinatized coverings of degree $n$
over $B$ chosen with the uniform distribution.
More generally, a {\em model} over a graph, $B$, is a collection of
probability spaces, $\{\cC_n(B)\}_{n\in N}$, 
defined for $n\in N$ where $N\subset\naturals$ is an
infinite subset, and where the atoms of each $\cC_n(B)$ are elements
of $\Coord_n(B)$.
There are a number of models related to the permutation model,
which are generalizations of the models of \cite{friedman_alon},
that we call {\em our basic models} and are defined in Article~I;
let us give a rough description.

All of {\em our basic models} are {\em edge independent}, meaning that
for any orientation $\Eor_B\subset\Edir_B$, the values of 
the permutation assignment, $\sigma$, on $\Eor_B$ are independent
of one another (of course, $\sigma(\iota_G e)=(\sigma(e))^{-1}$,
so $\sigma$ is determined by its values on any orientation
$\Eor_B$); for edge independent models, it suffices to specify
the ($\cS_n$-valued)
random variable $\sigma(e)$ for each $e$ in $\Eor_B$ or $\Edir_B$.
The permutation model can be alternatively described as the 
edge independent model that assigns a uniformly chosen permutation
to each $e\in\Edir_B$ (which requires $B$ to have no half-loops);
the {\em full cycle} (or simply {\em cyclic}) model is the same, except
that if $e$ is a whole-loop then $\sigma(e)$ is chosen uniformly
among all permutations whose cyclic structure consists of a single
$n$-cycle.
If $B$ has half-loops, then we restrict $\cC_n(B)$ either to $n$ even
or $n$ odd and for each half-loop $e\in\Edir_B$ we
choose $\sigma(e)$ as follows: if $n$ is even we choose 
$\sigma(e)$ uniformly among all perfect matchings,
i.e., involutions (maps equal to their inverse) with no fixed points;
if $n$ is odd then we choose $\sigma(e)$ uniformly among
all {\em nearly perfect matchings}, meaning involutions with one
fixed point.
We combine terms when $B$ has half-loops: for example,
the term {\em full cycle-involution} (or simply {\em cyclic-involution})
{\em model of odd degree over $B$} refers
to the model where the degree, $n$, is odd,
where $\sigma(e)$ follows the full cycle rule when $e$ is
not a half-loop, and where $\sigma(e)$ is a near perfect matching
when $e$ is a half-loop;
similarly for the {\em full cycle-involution} (or simply 
{\em cyclic-involution})
{\em model of even degree}
and the {\em permutation-involution model of even degree}
or {\em of odd degree}.

If $B$ is a graph, then a model, $\{\cC_n(B)\}_{n\in N}$, over $B$
may well have $N\ne \naturals$ (e.g., our basic models above when
$B$ has half-loops); in this case many formulas involving
the variable $n$ are only defined for $n\in N$.  For brevity, we
often do not explicitly write $n\in N$ in such formulas; 
for example we usually write
$$
\lim_{n\to\infty} \quad\mbox{to abbreviate}\quad
\lim_{n\in N,\ n\to\infty} \ .
$$
Also we often write simply $\cC_n(B)$ or $\{\cC_n(B)\}$ for
$\{\cC_n(B)\}_{n\in N}$ if confusion is unlikely to occur.

A graph is {\em pruned} if all its vertices are of degree at least
two (this differs from the more standard definition of {\em pruned} 
meaning that there are
no leaves).  If $w$ is any SNBC walk in a graph, $G$, then
we easily see that
$\ViSu_G(w)$ is necessarily pruned: i.e., any of its vertices must be
incident upon a whole-loop or two distinct edges
[note that a walk of length $k=1$ about a half-loop, $(v_0,e_1,v_1)$, by
definition, is not SNBC since $\iota_G e_k=e_1$].
It easily follows that $\ViSu_G(w)$ is contained in the graph
obtained from $G$ by repeatedly ``pruning any leaves''
(i.e., discarding any vertex of degree one and its incident edge)
from $G$.
Since our trace methods only concern (Hashimoto matrices and)
SNBC walks, it suffices to work with models $\cC_n(B)$ where
$B$ is pruned.
It is not hard to see that if $B$ is pruned and connected,
then $\ord(B)=0$ iff $B$ is a cycle,
and $\mu_1(B)>1$ iff $\chi(B)<0$;
this is formally proven in Article~III (Lemma~6.4).
Our theorems are not usually interesting unless $\mu_1(B)>\mu_1^{1/2}(B)$,
so we tend to restrict our main theorems
to the case $\mu_1(B)>1$ or, equivalently,
$\chi(B)<0$; some of our techniques work without these restrictions.

\subsection{Asymptotic Expansions}
\label{su_asymptotic_expansions}


A function $f\from\naturals\to\complex$ is a {\em polyexponential} if
it is a sum of functions $p(k)\mu^k$, where $p$ is a polynomial
and $\mu\in\complex$, with the convention
that for $\mu=0$ we understand $p(k)\mu^k$ to mean
any function that vanishes for sufficiently large $k$\footnote{
  This convention is used because then for any fixed matrix, $M$,
  any entry of $M^k$, as a function of $k$, is a polyexponential
  function of $k$; more specifically, the $\mu=0$ convention
  is due to the fact that a Jordan block of eigenvalue $0$ is
  nilpotent.
  }; we refer to the $\mu$
needed to express $f$ as the {\em exponents} or {\em bases} of $f$.
A function $f\from\naturals\to\complex$ is {\em of growth $\rho$}
for a $\rho\in\reals$ if $|f(k)|=o(1)(\rho+\epsilon)^k$ for any $\epsilon>0$.
A function $f\from\naturals\to\complex$ is $(B,\nu)$-bounded if it
is the sum of a function of growth $\nu$ plus a polyexponential function
whose bases are bounded by $\mu_1(B)$ (the Perron-Frobenius eigenvalue
of $H_B$); the {\em larger bases} of $f$ (with respect to $\nu$) are
those bases of the polyexponential function that are larger in
absolute value than $\nu$.
Moreover, such an $f$ is called {\em $(B,\nu)$-Ramanujan} if its
larger bases are all eigenvalues of $H_B$.

We say that a function $f=f(k,n)$ taking some subset of $\naturals^2$ to
$\complex$ has a 
{\em $(B,\nu)$-bounded expansion of order $r$} if for some
constant $C$ we have
\begin{equation}\label{eq_B_nu_defs_summ}
f(k,n) = c_0(k)+\cdots+c_{r-1}(k)+ O(1) c_r(k)/n^r,
\end{equation} 
whenever $f(k,n)$ is defined and $1\le k\le n^{1/2}/C$, where
for $0\le i\le r-1$, the $c_i(k)$ are $(B,\nu)$-bounded and $c_r(k)$
is of growth $\mu_1(B)$.
Furthermore, such an expansion is called {\em $(B,\nu)$-Ramanujan}
if for $0\le i\le r-1$, the $c_i(k)$ are {\em $(B,\nu)$-Ramanujan}.

Typically our functions $f(k,n)$ as in
\eqref{eq_B_nu_defs_summ} are defined for all $k\in\naturals$
and $n\in N$ for an infinite set $N\subset\naturals$ representing
the possible degrees of our random covering maps in the model
$\{\cC_n(B)\}_{n\in N}$ at hand.

\subsection{Tangles}
\label{su_tangles}

A {\em $(\ge\nu)$-tangle} is any 
connected graph, $\psi$, with $\mu_1(\psi)\ge\nu$,
where $\mu_1(\psi)$ denotes the Perron-Frobenius eigenvalue of $H_B$;
a {\em $(\ge\nu,<r)$-tangle} is any $(\ge\nu)$-tangle of order less than
$r$;
similarly for $(>\nu)$-tangles, i.e.,
$\psi$ satisfying the weak inequality $\mu_1(\psi)>\nu$,
and for $(>\nu,r)$-tangles.
We use ${\rm TangleFree}(\ge\nu,<r)$ to denote those graphs that don't
contain a subgraph that is $(\ge\nu,<r)$-tangle, and
${\rm HasTangles}(\ge\nu,<r)$ for those that do; we
never use $(>\nu)$-tangles in defining TangleFree and HasTangles,
for the technical reason
(see Article~III or Lemma~9.2 of \cite{friedman_alon}) that
for $\nu>1$ and any $r\in\naturals$
that there are only finitely many 
$(\ge\nu,<r)$-tangles, up to isomorphism, that are minimal
with respect to inclusion\footnote{
  By contrast, there are infinitely many minimal $(>\nu,<r)$-tangles
  for some values of $\nu>1$ and $r$: indeed, consider any connected pruned
  graph $\psi$, and set $r=\ord(\psi)+2$, $\nu=\mu_1(\psi)$.  Then if
  we fix two vertices in $\psi$ and let $\psi_s$ be the graph that is
  $\psi$ with an additional edge of length $s$ between these two 
  vertices, then $\psi_s$ is an $(>\nu,<r)$-tangle.  However, if
  $\psi'$ is $\psi$ with any single edge deleted, and $\psi'_s$ is 
  $\psi_s$ with this edge deleted, then one can show that
  $\mu_1(\psi'_s)<\nu$ for $s$ sufficiently large.  It follows that
  for $s$ sufficiently large, $\psi_s$ are minimal $(>\nu,<r)$-tangles.
}.

\subsection{$B$-Graphs, Ordered Graphs, and Strongly Algebraic Models}
\label{su_B_ordered_strongly_alg}

An {\em ordered graph}, $G^\og$, is a graph, $G$, endowed with an
{\em ordering}, meaning
an orientation (i.e., $\iota_G$-orbit representatives), 
$\Eor_G\subset\Edir_G$, 
and total orderings of $V_G$ and $E_G$;
a walk, $w=(v_0,\ldots,e_k,v_k)$ in a graph endows $\ViSu(w)$ with a
{\em first-encountered} ordering:
namely, $v\le v'$ if the first occurrence of $v$ comes before that
of $v'$ in the sequence $v_0,v_1,\ldots,v_k$,
similarly for $e\le e'$, and we orient each edge in the
order in which it is first traversed (some edges may be traversed
in only one direction).
We use $\ViSu^\og(w)$ to refer to $\ViSu(w)$ with this ordering.

A {\em morphism} $G^\og\to H^\og$ of ordered graphs is a morphism
$G\to H$ that respects the ordering in the evident fashion.
We are mostly interested in {\em isomorphisms} of ordered graphs;
we easily see that any isomorphism $G^\og\to G^\og$ must be the
identity morphism; it follows that if $G^\og$ and $H^\og$ are
isomorphic, then there is a unique isomorphism $G^\og\to H^\og$.

If $B$ is a graph, then a $B$-graph, $G_\Bg$, is a graph $G$ endowed 
with a map $G\to B$ (its {\em $B$-graph} structure).
A {\em morphism} $G_\Bg\to H_\Bg$ of $B$-graphs is a morphism
$G\to H$ that respects the $B$-structures in the evident sense.
An {\em ordered $B$-graph}, $G^\og_\Bg$, is a graph endowed with
both an ordering and a $B$-graph structure; a morphism of
ordered $B$-graphs is a morphism of the underlying graphs that
respects both the ordering and $B$-graph structures.
If $w$ is a walk in a $B$-graph, $G_\Bg$, we use $\ViSu_\Bg(w)$ to denote
$\ViSu(w)$ with the $B$-graph structure it inherits from $G$ in
the evident sense; we use $\ViSu_\Bg^\og(w)$ to denote
$\ViSu_\Bg(w)$ with its first-encountered ordering.

At times we drop the superscript $\,^\og$ and the subscript $\,_\Bg$;
for example, we write $G\in\Coord_n(B)$ instead of $G_\Bg\in\cC_n(B)$
(despite the fact that we constantly utilize
the $B$-graph structure on elements of
$\Coord_n(B)$).

A $B$-graph $G_\Bg$ is {\em covering} or {\'etale} if its structure
map $G\to B$ is.

If $\pi\from S\to B$ is a $B$-graph, we use
$\mec a=\mec a_{S_\Bg}$ to denote the vector
$\Edir_B\to\integers_{\ge 0}$ given by
$a_{S_\Bg}(e) = \# \pi^{-1}(e)$;
since $a_{S_\Bg}(\iota_B e) = a_{S_\Bg}(e)$ for all $e\in\Edir_B$,
we sometimes view $\mec a$ as a function $E_B\to\integers_{\ge 0}$, i.e.,
as the function taking $\{e,\iota_B e\}$ to 
$a_{S_\Bg}(e)=a_{S_\Bg}(\iota_B e)$.
We similarly define $\mec b_{S_\Bg}\from V_B\to\integers_{\ge 0}$ by
setting $b_{S_\Bg}(v) = \#\pi^{-1}(v)$.
If $w$ is a walk in a $B$-graph, we set $\mec a_w$ to be
$\mec a_{S_\Bg}$ where $S_\Bg=\ViSu_\Bg(w)$, and similarly for $\mec b_w$.
We refer to $\mec a,\mec b$ (in either context) as
{\em $B$-fibre counting functions}.

If $S_\Bg^\og$ is an ordered $B$-graph and $G_\Bg$ is a $B$-graph, we 
use $[S_\Bg^\og]\cap G_\Bg$ to denote the set of ordered graphs ${G'}_\Bg^\og$
such that $G'_\Bg\subset G_\Bg$ and ${G'}_\Bg^\og\isom S_\Bg^\og$
(as ordered $B$-graphs); this set is naturally identified with the
set of injective morphisms $S_\Bg\to G_\Bg$, and the cardinality of these
sets is independent of the ordering on $S_\Bg^\og$.


A $B$-graph, $S_\Bg$, or an ordered $B$-graph, $S_\Bg^\og$,
{\em occurs in a model $\{\cC_n(B)\}_{n\in N}$}
if for all sufficiently large
$n\in N$, $S_\Bg$ is isomorphic to a $B$-subgraph of some element
of $\cC_n(B)$; similary a graph, $S$, {\em occurs in 
$\{\cC_n(B)\}_{n\in N}$} if it can be endowed with a $B$-graph
structure, $S_\Bg$, that occurs in 
$\{\cC_n(B)\}_{n\in N}$.

A model $\{\cC_n(B)\}_{n\in N}$ of coverings of $B$ is {\em strongly
algebraic} if
\begin{enumerate}
\item for each $r\in\naturals$
there is a function, $g=g(k)$, of growth $\mu_1(B)$
such that if $k\le n/4$ we have
\begin{equation}\label{eq_algebraic_order_bound}
\EE_{G\in\cC_n(B)}[ \snbc_{\ge r}(G,k)] \le
g(k)/n^r
\end{equation}
where $\snbc_{\ge r}(G,k)$ is the number of SNBC walks of length
$k$ in $G$ whose visited subgraph is of order at least $r$;
\item
for any $r$ there exists
a function $g$ of growth $1$ and real $C>0$ such that the following
holds:
for any ordered $B$-graph, $S_\Bg^\og$, that is pruned and of
order less than $r$,
\begin{enumerate}
\item
if $S_\Bg$ occurs in $\cC_n(B)$, then for
$1\le \#\Edir_S\le n^{1/2}/C$,
\begin{equation}\label{eq_expansion_S}
\EE_{G\in\cC_n(B)}\Bigl[ \#\bigl([S_\Bg^\og]\cap G\bigr) \Bigr]
=
c_0 + \cdots + c_{r-1}/n^{r-1}
+ O(1) g(\# E_S) /n^r
\end{equation} 
where the $O(1)$ term is bounded in absolute value by $C$
(and therefore independent of $n$ and $S_\Bg$), and
where $c_i=c_i(S_\Bg)\in\reals$ such that
$c_i$ is $0$ if $i<\ord(S)$ and $c_i>0$ for $i=\ord(S)$;
and
\item
if $S_\Bg$ does not occur in $\cC_n(B)$, then for any
$n$ with $\#\Edir_S\le n^{1/2}/C$,
\begin{equation}\label{eq_zero_S_in_G}
\EE_{G\in\cC_n(B)}\Bigl[ \#\bigl([S_\Bg^\og]\cap G\bigr) \Bigr]
= 0 
\end{equation} 
(or, equivalently, no graph in $\cC_n(B)$ has a $B$-subgraph isomorphic to
$S_\Bg^\og$);
\end{enumerate}
\item
$c_0=c_0(S_\Bg)$ equals $1$ if $S$ is a cycle (i.e., $\ord(S)=0$ and
$S$ is connected) that occurs in $\cC_n(B)$;
\item
$S_\Bg$ occurs in $\cC_n(B)$ iff $S_\Bg$ is an \'etale $B$-graph
and $S$ has no half-loops; and
\item
there exist
polynomials $p_i=p_i(\mec a,\mec b)$ such that $p_0=1$
(i.e., identically 1), and for every
\'etale $B$-graph, $S_\Bg^\og$, we have that
\begin{equation}\label{eq_strongly_algebraic}
c_{\ord(S)+i}(S_\Bg) = p_i(\mec a_{S_\Bg},\mec b_{S_\Bg}) \ .
\end{equation}
\end{enumerate}
Notice that condition~(3), regarding $S$ that are cycles, is implied
by conditions~(4) and~(5); we leave in condition~(3) since this makes the
definition of {\em algebraic} (below) simpler.
Notice that \eqref{eq_expansion_S} and \eqref{eq_strongly_algebraic}
are the main reasons that we work with
ordered $B$-graphs: indeed, the coefficients depend only on
the $B$-fibre counting function $\mec a,\mec b$, which 
depend on the structure of
$S_\Bg^\og$ as a $B$-graph; this is not true if we don't work with
ordered graphs: i.e.,
\eqref{eq_expansion_S} fails to
hold if we replace $[S_\Bg^\og]$
with $[S_\Bg]$ (when $S_\Bg$ has nontrivial automorphisms), where
$[S_\Bg]\cap G$ refers to the number of $B$-subgraphs of $G$ isomorphic
to $S_\Bg$; the reason is that
$$
\#[S_\Bg^\og]\cap G_\Bg = \bigl( \#{\rm Aut}(S_\Bg)\bigr)
\bigl( \#[S_\Bg]\cap G_\Bg \bigr)
$$
where ${\rm Aut}(S_\Bg)$ is the group of automorphisms of $S_\Bg$, 
and it is $[S_\Bg^\og]\cap G_\Bg$ rather than $[S_\Bg]\cap G_\Bg$
that turns out to have the ``better'' properties;
see Section~6 of Article~I for examples.
Ordered graphs are convenient to use for a number of other reasons.

\ignore{
\myDeleteNote{Stuff deleted here and below on September 13, 2018.}
}

\subsection{Homotopy Type}

The homotopy type of a walk and of an ordered subgraph are defined
by {\em suppressing} its ``uninteresting'' vertices of degree two;
examples are given in Section~6 of Article~I.
Here is how we make this precise.

A {\em bead} in a graph is a vertex of degree two that is not
incident upon a self-loop.
Let $S$ be a graph and $V'\subset V_S$ be a {\em proper bead subset} of 
$V_S$,
meaning that $V'$ consists only of beads of $V$,
and that no connected component of $S$ has all its vertices in
$V'$ (this can only happen for connected components of $S$ that
are cycles);
we define the {\em bead suppression} $S/V'$ to be the following
graph: (1) its
vertex set $V_{S/V'}$
is $V''=V_S\setminus V'$, (2) its directed edges, $\Edir_{S/V'}$ consist
of
the {\em $V$'-beaded paths}, i.e., non-backtracking walks
in $S$ between elements of $V''$ whose intermediate vertices lie in $V'$,
(3) $t_{S/V'}$ and $h_{S/V'}$ give the first and last vertex of
the beaded path, and (4) $\iota_{S/V'}$ takes a beaded path
to its reverse walk
(i.e., takes $(v_0,e_1,\ldots,v_k)$ to
$(v_k,\iota_S e_k,\ldots,\iota_S e_1,v_0)$).
One can recover $S$ from the suppression $S/V'$ for pedantic reasons,
since we have defined its directed edges to be beaded paths of $S$.
If $S^\og=\ViSu^\og(w)$ where $w$ is a non-backtracking walk,
then the ordering of $S$ can be inferred by the naturally
corresponding order on $S/V'$, and we use $S^\og/V'$ to denote
$S/V'$ with this ordering.

Let $w$ be a non-backtracking walk in a graph, and 
$S^\og=\ViSu^\og(w)$ its visited
subgraph; the {\em reduction} of $w$ is the ordered graph,
$R^\og$, denoted $S^\og/V'$,
whose underlying graph is
$S/V'$ where $V'$ is the set of beads of $S$ except
the first and last vertices of $w$ (if one or both are beads),
and whose ordering is naturally arises from that on $S^\og$;
the {\em edge lengths} of $w$ is the function $E_{S/V'}\to\naturals$
taking an edge of $S/V'$ to the length of the beaded path it represents
in $S$;
we say that $w$ is {\em of homotopy type} $T^\og$ for any ordered
graph $T^\og$ that is isomorphic to $S^\og/V'$; in this case
the lengths of $S^\og/V'$ naturally give lengths $E_T\to\naturals$
by the unique isomorphism from $T^\og$ to $S^\og/V'$.
If $S^\og$ is the visited subgraph of a non-backtracking walk,
we define the reduction, homotopy type, and edge-lengths of $S^\og$ to
be that of the walk, since these notions depend only on $S^\og$ and
not the particular walk.

If $T$ is a graph and $\mec k\from E_T\to\naturals$ a function, then
we use $\VLG(T,\mec k)$ (for {\em variable-length graph}) to denote
any graph obtained from $T$ by gluing in a path of length $k(e)$
for each $e\in E_T$.  If $S^\og$ is of homotopy type $T^\og$
and $\mec k\from E_T\to \naturals$ its edge lengths,
then $\VLG(T,\mec k)$ is isomorphic to $S$ (as a graph).
Hence the construction of variable-length graphs is a sort of
inverse to bead suppression.

If $T^\og$ is an ordering on $T$ that arises as the first encountered
ordering of a non-backtracking walk on $T$ (whose visited subgraph
is all of $T$), then this ordering gives rise to a natural
ordering on $\VLG(T,\mec k)$ that we denote $\VLG^\og(T^\og,\mec k)$.
Again, this ordering on the variable-length graph is a sort of
inverse to bead suppression on ordered graphs.

\subsection{$B$-graphs and Wordings}

If $w_B=(v_0,e_1,\ldots,e_k,v_k)$ with $k\ge 1$ is a walk in a graph
$B$, then we can identify
$w_B$ with the string $e_1,e_2,\ldots,e_k$ over the alphabet
$\Edir_B$.
For technical reasons, the definitions below of
a {\em $B$-wording} and 
the {\em induced wording}, are given as strings over $\Edir_B$ rather
than the full alternating string of vertices and directed edges.
The reason is that 
doing this gives the correct notion of the {\em eigenvalues} of
an algebraic model (defined below).

Let $w$ be a non-backtracking walk in a $B$-graph, whose reduction
is $S^\og/V'$, and let
$S_\Bg^\og=\ViSu_\Bg^\og$.
Then the {\em wording induced by $w$} on $S^\og/V'$ is
the map $W$ from $\Edir_{S/V'}$ to strings in $\Edir_B$
of positive length, 
taking a
directed edge $e\in\Edir_{S/V'}$ to the string of $\Edir_B$ edges
in the non-backtracking walk in $B$
that lies under the walk in $S$ that it represents.
Abstractly, we say that a {\em $B$-wording} of a graph $T$
is a map $W$ from $\Edir_T$ to words over the alphabet
$\Edir_B$ that represent (the directed edges of)
non-backtracking walks in $B$ such that
(1) $W(\iota_T e)$ is the reverse word (corresponding to
the reverse walk) in $B$ of $W(e)$, 
(2) if $e\in\Edir_T$ is a half-loop, then $W(e)$ is of length one
whose single letter is a half-loop, and
(3) the tail of the first directed edge in $W(e)$ 
(corresponding to the first vertex in the associated walk in $B$)
depends only on $t_T e$;
the {\em edge-lengths} of $W$ is the function $E_T\to\naturals$
taking $e$ to the length of $W(e)$.
[Hence the wording induced by $w$ above is, indeed, a $B$-wording.]

Given a graph, $T$, and a $B$-wording $W$, there is a $B$-graph,
unique up to isomorphism, whose underlying graph is $\VLG(T,\mec k)$
where $\mec k$ is the edge-lengths of $W$, and where the $B$-graph
structure maps the non-backtracking walk in $\VLG(T,\mec k)$
corresponding to an $e\in\Edir_T$ to the non-backtracking walk in $B$
given by $W(e)$.
We denote any such $B$-graph by $\VLG(T,W)$; again this is
a sort of inverse to starting with a non-backtracking walk
and producing the wording it induces on its visited subgraph.

Notice that if $S_\Bg^\og=\VLG(T^\og,W)$ for a $B$-wording, $W$,
then the $B$-fibre counting functions
$\mec a_{S_\Bg}$ and $\mec b_{S_\Bg}$ can be
inferred from $W$, and we may therefore write $\mec a_W$ and
$\mec b_W$.

\subsection{Algebraic Models}

By a $B$-type we mean a pair $T^{\rm type}=(T,\cR)$ consisting
of a graph, $T$, and a map from $\Edir_T$ to the set
of regular languages over the alphabet $\Edir_B$ (in the sense of regular
language theory) such that
(1) all words in $\cR(e)$ are positive length strings corresponding to
non-backtracking walks in $B$, 
(2) if for $e\in\Edir_T$ we have $w=e_1\ldots e_k\in\cR(e)$,
then $w^R\eqdef \iota_B e_k\ldots\iota_B e_1$ lies in $\cR(\iota_T e)$,
and (3) if $W\from \Edir_T\to(\Edir_B)^*$ (where $(\Edir_B)^*$ is
the set of strings over $\Edir_B$) satisfies
$W(e)\in\cR(e)$ and $W(\iota_T e)=W(e)^R$ for all $e\in \Edir_T$,
then $W$ is a $B$-wording.
A $B$-wording $W$ of $T$ is {\em of type $T^{\rm type}$} if
$W(e)\in\cR(e)$ for each $e\in\Edir_T$.

Let $\cC_n(B)$ be a model that satisfies (1)--(3) of the definition
of strongly algebraic.
If $\cT$ a subset of $B$-graphs,
we say that the model is {\em algebraic restricted to $\cT$}
if 
either all $S_\Bg\in\cT$ occur in $\cC_n(B)$ or they all do not,
and if so
there are polynomials $p_0,p_1,\ldots$ such that
$c_i(S_\Bg)=p_i(S_\Bg)$ for any $S_\Bg\in\cT$. 
We say that $\cC_n(B)$ is {\em algebraic} if 
\begin{enumerate}
\item
setting $h(k)$ to be
the number of $B$-graph isomorphism classes of \'etale $B$-graphs
$S_\Bg$ such that $S$ is a cycle of length $k$ and $S$ does
not occur in $\cC_n(B)$, we have that 
$h$ is a function of growth $(d-1)^{1/2}$; and
\item
for any
pruned, ordered graph, $T^\og$, there is a finite number of
$B$-types, $T_j^{\rm type}=(T^\og,\cR_j)$, $j=1,\ldots,s$, 
such that (1) any $B$-wording, $W$, of $T$ belongs to exactly one
$\cR_j$, and
(2) $\cC_n(B)$ is algebraic when restricted to $T_j^{\rm type}$.
\end{enumerate}

[In Article~I we show that
if instead each $B$-wording belong to 
{\em at least one} $B$-type $T_j^{\rm type}$, then one can choose a
another set of
$B$-types that satisfy (2) and where each $B$-wording belongs
to {\em a unique} $B$-type;
however, the uniqueness
is ultimately needed in our proofs,
so we use uniqueness in our definition of algebraic.]

We remark that one can say that a walk, $w$, in a $B$-graph,
or an ordered $B$-graphs, $S_\Bg^\og$, is of {\em homotopy type $T^\og$},
but when $T$ has non-trivial automorphism one {\em cannot} say
that is of $B$-type $(T,\cR)$ unless---for example---one orders
$T$ and speaks of an {\em ordered $B$-type}, $(T^\og,\cR)$.
[This will be of concern only in Article~II.]

We define the {\em eigenvalues} of a regular language, $R$, to be the minimal
set $\mu_1,\ldots,\mu_m$ such that for any $k\ge 1$,
the number of words of length $k$ in the language
is given as
$$
\sum_{i=1}^m p_i(k)\mu_i^k
$$
for some polynomials $p_i=p_i(k)$, with the convention that
if $\mu_i=0$ then $p_i(k)\mu_i^k$ refers to any function that 
vanishes for $k$ sufficiently large (the reason for this is that
a Jordan block of eigenvalue $0$ is a nilpotent matrix).
Similarly, we define the eigenvalues of a $B$-type $T^{\rm type}=(T,\cR)$
as the union of all the eigenvalues of the $\cR(e)$.
Similarly a {\em set of eigenvalues} of a graph, $T$
(respectively, an algebraic model, $\cC_n(B)$)
is
any set containing the eigenvalues containing the eigenvalues
of some choice of $B$-types used in the definition of algebraic
for $T$-wordings (respectively, for $T$-wordings for all $T$).

[In Article~V we prove that all of our basic models are algebraic;
some of our basic models, such as the
permutation-involution model and the cyclic models, are not
strongly algebraic.]

We remark that a homotopy type, $T^\og$,
of a non-backtracking walk, can only have beads as its first or last 
vertices; however, in the definition of algebraic we require
a condition on {\em all pruned graphs}, $T$, 
which includes $T$ that may have many beads and may not be connected;
this is needed
when we define homotopy types of pairs in Article~II.

\subsection{SNBC Counting Functions}

If $T^\og$ is an ordered graph and $\mec k\from E_T\to\naturals$, 
we use $\SNBC(T^\og,\mec k;G,k)$ to denote the set of SNBC walks in $G$
of length $k$ and of homotopy type $T^\og$ and edge lengths $\mec k$.
We similarly define
$$
\SNBC(T^\og,\ge\bec\xi;G,k) \eqdef 
\bigcup_{\mec k\ge\bec\xi} \SNBC(T^\og,\mec k;G,k)
$$
where $\mec k\ge\bec\xi$ means that $k(e)\ge\xi(e)$ for all $e\in E_T$.
We denote the cardinality of these sets by replacing $\SNBC$ with
$\snbc$;
we call $\snbc(T^\og,\ge\bec\xi;G,k)$ the set of 
{\em $\bec\xi$-certified
traces of homotopy type $T^\og$ of length $k$ in $G$};
in Article~III we will refer to certain $\bec\xi$ as {\em certificates}.

\section{Main Theorems in this Article}
\label{se_main_thms}

In this section we state the main theorems that we prove.  First we state
the simplest theorem.

\begin{theorem}\label{th_main_illustration}
Let $B$ be a graph, and $\{\cC_n(B)\}_{n\in N}$ an algebraic model over
$B$.  Let 
$T^\og$ be an ordered graph,
and let
$$
\nu = \max\bigl( \mu_1^{1/2}(B), \mu_1(T) \bigr).
$$
Then 
$$
f(k,n)\eqdef \EE_{G\in\cC_n(B)}[\snbc(T^\og;G;k)]
$$
has a $(B,\nu)$-bounded expansion
$$
c_0(k)+\cdots+c_{r-1}(k)/n^{r-1}+ O(1) c_r(k)/n^r
$$
to any order $r$, and the
bases of the coefficients of the expansion is a subset of 
any set of eigenvalues for the model $\cC_n(B)$.
\end{theorem}
We recall from the discussion regarding \eqref{eq_B_nu_defs_summ}
that the above expansion holds for $k,n$ where $f(k,n)$ is defined and
$1\le k\le n^{1/2}/C$ for some constant $C$ depending on $r$;
since $f(k,n)$ is only defined for $n\in N$ (and any $k$), the above expansion
holds for all $(k,n)$ with $n\in N$ and $1\le k\le n^{1/2}/C$.

The proof of the above theorem is a straightforward generalization of
the theorems \cite{friedman_random_graphs}.
The difference in this article is that we ``factor'' this proof into a
number of independent parts, and each part is stated in as general
a form as is reasonably possible.

We also note that Theorem~\ref{th_main_illustration} shows the
limits of the trace methods of \cite{friedman_random_graphs}:
when $T$ is a bouquet of $m$ whole-loops and $\mu_1(T)=2m-1$
exceeds $d-1$ (for any integer $d\ge 2$), 
then Theorem~\ref{th_main_illustration} can no longer
produce $(B,\nu)$-bounded coefficients with $\nu$ arbitrarily close
to $(d-1)^{1/2}$.  For this reason \cite{friedman_random_graphs} 
limits itself to $r$ of order $d^{1/2}$; Puder
\cite{puder} gets a better bound; in our language and methods, 
Puder's improvement is due to the fact that $c_i$
is of growth
$$
\max\bigl( (d-1)^{1/2}, 2i+1 \bigr)
$$
for all $i$, which can be obtained
from Theorem~\ref{th_main_illustration}, and which Puder
exploits in the range where $i$ is greater than order $d^{1/2}$
(up to $2i+1\le d-1$).

The next two theorems are successively stronger.  However, the proofs
of these theorems are easy adaptations of the proof of 
Theorem~\ref{th_main_illustration}, except for the additional
terminology needed regarding homotopy type; hence the reader can read almost
the entire article with only
Theorem~\ref{th_main_illustration} in mind.
Unlike Theorem~\ref{th_main_illustration}, the motivation for the
next theorem will not be clear until we discuss certified traces
in Article~III (see also Appendix~A of Article~I).

\begin{theorem}\label{th_main_certified_walks}
Let $B$ be a graph, and $\{\cC_n(B)\}_{n\in N}$ an algebraic model over
$B$.  Let 
$T^\og$ be an ordered graph, let
$\bec\cert\from E_T\to\naturals$ be a function, 
and let
$$
\nu = \max\Bigl( \mu_1^{1/2}(B), \mu_1\bigl(\VLG(T,\bec\cert)\bigr) \Bigr).
$$
Then for any $r\ge 1$ we have
\begin{equation}\label{eq_snbc_sum_first_main_thm}
f(k,n)\eqdef \EE_{G\in\cC_n(B)}[\snbc(T^\og,\ge\bec\cert;G;k)]
\end{equation} 
has a $(B,\nu)$-bounded expansion 
$$
c_0(k)+\cdots+c_{r-1}(k)/n^{r-1}+ O(1) c_r(k)/n^r
$$
to order $r$, where the
bases of the coefficients of the expansion is a subset of any set of
eigenvalues of $T^\og$ in the model $\cC_n(B)$;
furthermore $c_i(k)=0$ for
$i<\ord(T)$.
\end{theorem}

The proof of the theorem above is easily modified to prove the
next theorem, although the terminology
is more cumbersome.  It is used in Article~III when certified traces are
multiplied by indicator functions of tangles; see Article~I for an
overview of its
importance, which is akin to
equation~(39) in the proof of Theorem~9.3 of \cite{friedman_alon}.

\begin{theorem}\label{th_main_certified_pairs}
Let $\{\cC_n(B)\}_{n\in N}$ be an algebraic model over a graph $B$.
Let $T^\og$ be an ordered graph, let $\bec\xi\from E_T\to\naturals$ be
a function, and let
\begin{equation}\label{eq_nu_as_in_main_cert_pairs}
\nu = \max\Bigl( \mu_1^{1/2}(B), \mu_1\bigl(\VLG(T,\bec\cert)\bigr) \Bigr).
\end{equation}
Let $\psi_\Bg^\og$ be any pruned ordered $B$-graph.
Then for any
$r\ge 1$ we have
\begin{equation}\label{eq_subgraphs_times_walks}
\EE_{G\in\cC_n(B)}[ 
(\#[\psi_\Bg^\og]\cap G)
\snbc(T^\og;\ge\bec\xi,G,k) ]
\end{equation} 
has a $(B,\nu)$-bounded expansion of order $r$,
$$
c_0(k)+\cdots+c_{r-1}(k)/n^{r-1}+ O(1) c_r(k)/n^r;
$$
the bases of
the coefficients in the expansion are some subset of any set of
eigenvalues of the model, and $c_i(k)=0$ for $i$ less than the
order of all $B$-graphs that occur in $\cC_n(B)$ and contain both a walk of 
homotopy type $T^\og$
and a subgraph isomorphic to $\psi_\Bg$.
\end{theorem}

Note that Theorem~\ref{th_main_certified_pairs} reduces to
of Theorem~\ref{th_main_certified_walks} in the case where
$\psi_\Bg^\og$ is the empty graph (which has a unique
ordering and $B$-graph structure).

\ignore{\tiny\red

\subsection{First Step: SNBC Counts for $T^{\rm type}$ Versus $T^\og$}

The first step the proofs of our main theorems is to
replace $\snbc(T^\og,\ge\bec\xi;G_\Bg,k)$ by the analogous count
with $T^\og$ replaced by a $B$-type $T^{\rm type}=(T^\og,\cR)$.

\begin{definition}
Let $B$ be a graph, $T^\og$ an ordered graph,
and $T^{\rm type}=(T^\og,\cR)$ an ordered $B$-type.  We define
$$
\SNBC(T^{\rm type},\ge\bec\xi;G_\Bg,k)
$$
to be the subset of elements of (i.e., walks in)
$\SNBC(T^\og,\ge\bec\xi;G_\Bg,k)$ whose wording on $T^\og$
belongs to $T^{\rm type}=(T^\og,\cR)$.  
We use $\snbc$ instead of $\SNBC$ to denote
the cardinality of this set.
\end{definition}

By the definition of algebraic, for each $T$ there are finitely
many ordered $B$-types
$T^{\rm type}_j=(T^\og,\cR_j)$ such that for any $B$-graph, $G$ we have
\begin{equation}\label{eq_sum_over_ordered_types}
\snbc(T^\og,\ge\bec\xi;G_\Bg,k) = \sum_j 
\snbc(T^{\rm type}_j,\ge\bec\xi;G_\Bg,k) .
\end{equation} 
Since the right-hand-side is a finite sum
(and since a finite sum of $(B,\nu)$-bounded expansions is again such an
expansion), 
it suffices to prove these our main theorems where
$\snbc(T^\og,\ge\bec\xi;G,k)$
is replaced with
$\snbc(T^{\rm type},\ge\bec\xi;G,k)$ for an ordered $B$-type
$T^{\rm type}=(T^\og,\cR)$.

}

\section{The Length-Multiplicity Formula}
\label{se_length_mult}

In this section we give the first step in 
the proof of Theorem~\ref{th_main_certified_walks}, which we call the
{\em length-multiplicity formula}.

Assume the hypotheses of Theorems~\ref{th_main_certified_walks}.
We wish to prove that
\begin{equation}\label{eq_expected_T_veck_k}
\EE_{G\in\cC_n(B)}\bigl[\snbc(T^{\rm type},\ge\bec\xi;G;k)\bigr] 
=
\sum_{\mec k\ge \bec\xi}
\EE_{G\in\cC_n(B)}\bigl[\snbc(T^{\rm type},\mec k;G;k)\bigr]
\end{equation}
has a $(B,\nu)$-bounded asymptotic expansion to any order $r$.

\subsection{The Difficulty With Walks of Order One or More}

Our strategy, on the finest level,
is to fix an ordered $B$-graph, $S_\Bg^\og$,
and consider the $G\in\cC_n(B)$
expected number of walks, $w$, in $\SNBC(G,k)$
such that $\ViSu_\Bg^\og(w)$ is isomorphic to $S_\Bg^\og$.
The expression
$$
\EE_{G\in\cC_n(B)}\Bigl[ \#\bigl([S_\Bg^\og]\cap G\bigr) \Bigr],
$$
tells us how many subgraphs of $G_\Bg$, endowed with an ordering,
are isomorphic to $S_\Bg^\og$.  For any ordered subgraph,
${\tilde S}_\Bg^\og$, in $[S_\Bg^\og]\cap G$, the number of
SNBC walks in $G$ with this visited subgraph is a function
\begin{equation}\label{eq_define_vis_snbc}
\vis({\tilde S}_\Bg^\og,G,k) \eqdef
\# \bigl\{ w\in\SNBC(G,k) \ \bigm| 
\ \ViSu_\Bg^\og(w)={\tilde S}_\Bg^\og \bigr\}
\end{equation} 
To simplify matters we note that
\begin{equation}\label{eq_simplify_vis}
\vis({\tilde S}_\Bg^\og,G,k) = \vis(\tilde S^\og,k)
\end{equation} 
depends only on $\tilde S^\og$ and $k$, and not on $G$ or the $B$-structure
on $\tilde S_\Bg^\og$.
Unfortunately, this function $\vis(\tilde S^\og,k)$ is generally
a complicated function of $k$ for fixed ${\tilde S}^\og$.

\begin{example}
The only case where $\vis(\tilde S^\og,k)$ is very simple is
the case where $T^\og$ has order $0$;
therefore (see Section~6 of Article~I), $\tilde S$ is 
necessarily a cycle of some length $k_1$, and then
\begin{equation}\label{eq_h_ord_zero}
\vis(\tilde S^\og,k) =
\left\{ \begin{array}{ll} $1$ & \mbox{if $k_1$ divides $k$, and} \\
$0$ & \mbox{otherwise}
\end{array}\right.
\end{equation} 
assuming that ${\tilde S}_\Bg^\og$ is the visited subgraph of
some SNBC walk.
[Recall from Article~I that if $k_1$ divides $k$, then
there are $2k_1$ SNBC walks of length $k$ in a cycle, $C$, of length $k_1$,
but there are only $0$ or $1$ when $C$ is endowed with an ordering
(and $1$ when the ordering arises from some SNBC walk whose ordered visited 
subgraph is $C^\og$).]
\end{example}

\subsection{The Abstract Length-Multiplicity Formula}

Now we make some observations on $\vis(\tilde S^\og,k)$.

First, for any $S_\Bg^\og$,
taking expectations in \eqref{eq_define_vis_snbc}
over all $\tilde S_\Bg^\og$ isomorphic to $S_\Bg^\og$ shows 
that for any $n$ and $S_\Bg^\og$ we have
\begin{align}
\nonumber
&
\EE_{G\in\cC_n(B)}\Bigl[ \#\bigl\{
w\in\SNBC(G,k) \ \bigm| \ \ViSu_\Bg^\og(w)\isom S_\Bg^\og 
\bigr\} \Bigr] 
\\
\label{eq_expected_walks_fixed_vs}
& =
\vis(S^\og,k)
\,\EE_{G\in\cC_n(B)}\Bigl[ \#\bigl([S_\Bg^\og]\cap G\bigr) \Bigr] .
\end{align} 

\begin{example}
To get some intuition regarding
\eqref{eq_expected_walks_fixed_vs}, recall
(Section~6 of Article~I) that when $T^\og$ is the unique SNBC walk
homotopy type of order
zero, then we sum 
\eqref{eq_expected_walks_fixed_vs} over all $S_\Bg^\og$ of length
$k'$, and by \eqref{eq_h_ord_zero} it suffices to consider $k'$
dividing $k$.
We then use the fact that $\cC_n(B)$ is algebraic to expand the
expected value in 
\eqref{eq_expected_walks_fixed_vs} an asymptotic series.  For example,
if $\cC_n(B)$ is the permutation model and $B$ is a bouquet of
$d/2$ whole-loops (as in \cite{broder,friedman_random_graphs}), then
the first two terms are
$$
\EE_{G\in\cC_n(B)}\Bigl[ \#\bigl([S_\Bg^\og]\cap G\bigr) \Bigr]
=
1 + c_1(S_\Bg)/n + O(1)/n^2,
$$
where 
$$
c_1(S_\Bg) = \sum_{1\le j_1< j_2 \le d/2} a_S(f_{j_1}) a_S(f_{j_2})
$$
with $f_1,\ldots,f_{d/2}$ being any orientation of $B$.
We then sum this over all $S_\Bg^\og$ whose length is $k'$,
which we can identify with all SNBC walks in $B$ of length $k'$,
to obtain the first two terms in 
\eqref{eq_expected_T_veck_k} for this $T^\og$ of order $0$.
\end{example}

The {\em length-multiplicity} formula is an approach to
working with $\vis(S^\og,k)$ when $T^\og$ is a homotopy type of
order $1$ or more, used in \cite{friedman_random_graphs,friedman_alon}.

Let us introduce some useful definitions.

\begin{definition}\label{de_legal}
A walk, $w$, in an ordered graph, $S^\og$, is {\em legal (in $S^\og$)}
if it is SNBC (strictly non-backtracking closed)
and $\ViSu_S^\og(w)=S^\og$.  For such a walk we define
the {\em edge multiplicities of $w$ in $S$} as the vector 
$\mec m\from E_S\to\naturals$ where $m(e)$ is the number of times
that $e$ is traversed along $w$ in either direction.
For an ordered graph, $S^\og$, and function
$\mec m\from E_{S}\to\naturals$, let $\legal(S^\og,\mec m)$ denote the
number of legal walks, $w$, in $S^\og$ with edge multiplicities $\mec m$.
\end{definition}

The following simple observation is fundamental to our methods.

\begin{lemma}\label{le_fundamental_legal_formula}
Let $S^\og$ be an ordered graph that is the homotopy type of
some SNBC walk; let $T^\og$ be its homotopy type and let
$\mec k=\mec k_S$ be the edge lengths $E_T\to\naturals$ that
$S^\og$ induces on $T^\og$.
Then
\begin{equation}\label{eq_fundamental_num_walks}
\vis(S^\og,k)  = \sum_{\mec k_S\cdot \mec m=k} 
\legal(T^\og,\mec m).
\end{equation} 
\end{lemma}
The proof is simple, but a bit long to write down carefully.
\begin{proof}
By definition, $T^\og=S^\og/V'$ where $V'$ is the set of all beads
of $S$ except the first vertex of $S^\og$ (if it is a bead).
For each $e_T\in\Edir_T$, let ${\rm beaded}(e_T)$ denote the
$V'$-beaded path in $S$ corresponding to $e_T$, and let 
${\rm first}(e_T)\in\Edir_S$ be
the first directed edge in this beaded path.

Now we describe a simple 
one-to-one correspondence between legal walk in $S^\og$ and $T^\og$:
if $w_S$ is a non-backtracking walk in $S$, then whenever $w_S$
traverses ${\rm first}(e_T)$ for some $e_T$, then it must
immediately traverse the rest of ${\rm beaded}(e_T)$, whereupon
it ends on a vertex of $V_S/V'=V_T$.
This sets up a correspondence between non-backtracking walks in $S$
that begin and end in vertices of $V_T=V_S/V'$ and non-backtracking
walks in $T$.
We easily see that if $w_S$ corresponds to $w_T$, then
(1) $w_S$ is SNBC in $S$ iff $w_T$ is so in $T$, 
(2) $w_S$ visits all of $S$ iff $w_T$ does so in $T$, and
(3) the first encountered order of $w_S$ on $S$ is $S^\og$ iff
the same on $w_T$ on $T$ is $T^\og$.
Since legal walks in $S^\og$ begin and end on the first vertex of 
$S^\og$, which lies in $V_T=V_S\setminus V'$, it follows that
under this correspondence, $w_S$ is legal in $S^\og$
iff its corresponding walk, $w_T$, is legal in $T^\og$.
Furthermore, if the multiplicity of $w_T$ is
$\mec m\from E_T\to \naturals$, then the multiplicity of $w_S$
on each edge $E_S$ is the same as than on edge of $E_T$ to which
it corresponds.  Hence the length of $w_S$ equals
$\mec k\cdot\mec m$.
It follows that $w_S$ is length $k$ iff the corresponding walk
$w_T$ in $T$ has $\mec k\cdot\mec m=k$.  
Hence the sum over all $w_S$
that are legal walks in $S^\og$ length $k$ equals the sum
of all its corresponding legal walks $T^\og$ with 
$\mec k\cdot\mec m=k$; hence we conclude
\eqref{eq_fundamental_num_walks}
\end{proof}


[The same lemma holds for non-backtracking walks using the
more general correspondence of non-backtracking walks in $S$ and $T$, since
by definition the homotopy type of such walks never suppresses the
first or the last vertex.  Both \cite{broder,friedman_random_graphs}
worked with non-backtracking walks, and hence use this
modification of Lemma~\ref{le_fundamental_legal_formula}; on the other
hand,
this series of articles
and \cite{friedman_alon} work only with SNBC walks.]

In view of \eqref{eq_expected_walks_fixed_vs} 
and \eqref{eq_fundamental_num_walks}, we have
\begin{align}
\label{eq_expect_walks_isom_to_fixed}
& 
\EE_{G\in\cC_n(B)}\Bigl[\# \bigl\{ w\in \SNBC(G,k)  \ \bigm|  
\ \ViSu_\Bg^\og(w)\isom S_\Bg^\og \bigr\} \Bigr]
\\
\label{eq_individual_term_length_mult}
& = 
\EE_{G\in\cC_n(B)}\Bigl[ \#\bigl([S_\Bg^\og]\cap G\bigr) \Bigr] 
\sum_{\mec k_S\cdot \mec m=k} 
\legal(T^\og,\mec m).
\end{align} 

\begin{definition}\label{de_cT_length_set}
Let $T^\og$ be an ordered graph, and let $\cT$ be any class of
$B$-graphs of homotopy type $T^\og$ that is the union of isomorphism
classes of ordered $B$-graphs (i.e., if $S_\Bg^\og\in \cT$, then
any ordered $B$-graph isomorphic to $S_\Bg^\og$ lies in $\cT$ as 
well).  For any $\mec k\from E_T\to\naturals$, the
{\em subset of $\cT$ elements of lengths $\mec k$}, denoted
$\cT[\mec k]$, is the subset of $S_\Bg^\og$ whose edge lengths
$\mec k_S$ equal $\mec k$.
\end{definition}

If $\cT$ denotes any subclass of 
ordered $B$-graphs of homotopy type $T^\og$ which is a union
of isomorphism classes of ordered $B$-graphs, then
summing
the above equality over one representative in $\cT$ 
(for each class of ordered $B$-graphs) we have
\begin{align}
\label{eq_abstract_length_mult_lhs}
& \EE_{G\in\cC_n(B)}\Bigl[\# \bigl\{ w\in \SNBC(G,k)  \ \bigm|  
\ \ViSu_\Bg^\og(w)\in \cT \bigr\} \Bigr]
\\ 
\label{eq_abstract_length_mult_rhs}
& =
\sum_{\mec k\cdot \mec m=k} 
\left(
\sum_{[S_\Bg^\og]\in \cT[\mec k]}
\EE_{G\in\cC_n(B)}\Bigl[ \#\bigl([S_\Bg^\og]\cap G\bigr) \Bigr] 
\right) 
\legal(T^\og,\mec m)
\end{align} 
where the sum $[S_\Bg^\og]\in \cT[\mec k]$ means that we sum over one
representative $S_\Bg^\og$ of each isomorphism class of ordered
$B$-graphs.  
In fact, we shall apply this with a fairly broad set of values for
$\cT$, and 
our most general
{\em length-multiplicity} formula is based on the above formula.
Let us state this ``abstract'' length-multiplicity formula.

\begin{lemma}\label{le_abs_length_mult}
Let $B$ be any graph, $n\in\naturals$, and $\cC_n(B)$ a probability space
of $B$-graphs of degree $n$.
Let $T^\og$ be an ordered graph, and let 
$\cT$ be any class of
ordered $B$-graphs of homotopy type $T^\og$ that is the union of isomorphism
classes of ordered $B$-graphs.
Then for any $k,n\in\naturals$
$$
\EE_{G\in\cC_n(B)}\Bigl[\# \bigl\{ w\in \SNBC(G,k)  \ \bigm|  
\ \ViSu_\Bg^\og(w)\in \cT \bigr\} \Bigr]
= \sum_{\mec k\cdot\mec m=k} F_1(\mec k,n)F_2(\mec m),
$$
where 
\begin{equation}\label{eq_define_abs_F_1_2}
F_1(\mec k,n) \eqdef
\sum_{[S_\Bg^\og]\in\cT[\mec k]}
\EE_{G\in\cC_n(B)}\Bigl[\#[S_\Bg^\og]\cap G_\Bg\Bigr],\quad
F_2(\mec m)\eqdef \legal(T^\og,\mec m) \ .
\end{equation}
\end{lemma}
\begin{proof}
This is immediate from
\eqref{eq_abstract_length_mult_lhs} and
\eqref{eq_abstract_length_mult_rhs}.
\end{proof}

\subsection{The Length-Multiplicity Formula in Applications}

To prove the main expansions theorem in this article, such as for
\eqref{eq_expected_T_veck_k}, we will apply
Lemma~\ref{le_abs_length_mult} to some very special cases
that we now describe.

\begin{definition}
Let $B$ be a graph.  By an {\em ordered $B$-type} we mean
an ordered graph, $T^\og$, and a $B$-type $(T,\cR)$;
we use $T^{\rm Otype}=(T^\og,\cR)$ to denote an ordered $B$-type.
Let $w$ be a non-backtracking walk in some $B$-graph;
we say that $w$ or $S_\Bg^\og=\ViSu_\Bg^\og(w)$ 
is {\em of ordered $B$-type}
(or just {\em of type}) $T^{\rm Otype}=(T^\og,\cR)$, if $S_\Bg^\og$ is
of homotopy type $T^\og$, and the wording that $S_\Bg^\og$ induces on
$T$ is a wording in $\cR$.
\end{definition}
Notice that, in contrast to the above definition, it is more problematic
to say whether or not $S^\og_\Bg$ {\em is of homotopy type $(T,\cR)$}
when $T$ has nontrivial automorphisms:
without an ordering on $T$, the wording that $S_\Bg^\og$ induces on
$T$ depends on isomorphism of
the reduction of $S_\Bg^\og$ to $T$.
To prove Theorem~\ref{th_main_certified_walks} we will write
\eqref{eq_snbc_sum_first_main_thm} as a sum over ordered $B$-types by
partition all SNBC walks of homotopy type $T^\og$
according to their ordered $B$-type; this is why we will need
{\em ordered $B$-types} as opposed to $B$-types.

\begin{definition}
If $T^\og$ is an ordered graph and $\mec k\from E_T\to\naturals$ a function,
we let $\subgr_B^\og(T^\og,\mec k;G)$ be the number of ordered $B$-graphs
in $G$ of homotopy type $T^\og$ and 
edge lengths $\mec k$.
If, moreover,
$T^{\rm Otype}=(T^\og,\cR)$ is a ordered $B$-type, then we similarly
use $\subgr_B^\og(T^{\rm Otype},\mec k;G)$ to denote the number of
such ordered $B$-graphs that in addition are of type $T^{\rm Otype}$.
\end{definition}

We may also write
$$
\subgr_B^\og(T^\og,\mec k; G) = \# T^\og[\mec k]
$$
in the notation of Definition~\ref{de_cT_length_set}, understanding
$T^\og$ to refer to all ordered $B$-graphs of homotopy type
$T^\og$.  Similarly we write
$$
\subgr_B^\og(T^{\rm Oype},\mec k; G) = \# T^{\rm Otype}[\mec k]
$$
for any ordered $B$-type, $T^{\rm Otype}$.

\ignore{\tiny\red
Later on we will use the fact that 
\begin{equation}\label{eq_subgraphs_type_length}
\subgr_B^\og(T^\og,\mec k;G) = \sum_{[S_\Bg^\og]\in T^\og[\mec k]}
\bigl( \# [S_\Bg^\og]\cap G \bigr)
\end{equation} 
where $T^\og[\mec k]$ denotes the isomorphism classes of ordered $B$-graphs
of homotopy type $T^\og$ and edge lengths $\mec k\from E_T\to\naturals$;
furthermore, we will apply this in the context of an algebraic model,
where we have asymptotic expansions for the $G\in\cC_n(B)$ expected value
of $\# [S_\Bg^\og\cap G]$.  
For now we will ignore this structure and make some general remarks
that work with $\subgr_B^\og(T^\og,\mec k;G)$ in any context.

Note that
any walk in a $B$-graph is legal in and only in $\ViSu^\og_\Bg(w)$.
Furthermore, the legal walks in any ordered $B$-graph $S_\Bg^\og$ are
in one-to-one correspondence with the legal walks in the homotopy
reduction of $w$.
It follows that for any $B$-graph, $G_\Bg$, we have
$$
\snbc(T^\og,\mec k;G_\Bg,k)=
\subgr_B^\og(T^\og,\mec k;G)
\sum_{\mec m,\ {\rm s.t.}\ \mec k\cdot\mec m=k}
\legal(T^\og,\mec m) \ .
$$
Taking expected values yields the following lemma.
}

\begin{lemma}\label{le_length_mult}
Let $B$ be any graph, $n\in\naturals$, and $\cC_n(B)$ a probability space
of $B$-graphs of degree $n$.
For any $B$-type, $T^{\rm Otype}=(T^\og,\cR)$ 
and $k\in\naturals$ we have
\begin{equation}\label{eq_certified_type}
\EE_{G_\Bg\in\cC_n(B)}[ \snbc(T^{\rm Otype},\ge\bec\xi;G_\Bg,k) ]
= \sum_{\mec k\cdot\mec m=k, \ \mec k\ge \bec\xi} F_1(\mec k,n)F_2(\mec m),
\end{equation} 
where 
\begin{equation}\label{eq_define_F_1}
F_1(\mec k,n) \eqdef
\EE_{G\in\cC_n(B)}\Bigl[\subgr_B\bigl(T^{\rm Otype},\mec k;G_\Bg\bigr)\Bigr],
\quad
F_2(\mec m)\eqdef \legal(T^\og,\mec m) \ .
\end{equation}
Similarly the same holds with $T^{\rm Otype}$ replaced by $T^\og$.
\end{lemma}
\begin{proof}
In Lemma~\ref{le_abs_length_mult} we take
$\cT$ equal to all ordered $B$-subgraphs of $B$-type $T^{\rm Otype}$
and edge lengths $\mec k\ge\bec\xi$; this yields
\eqref{eq_certified_type} and
\eqref{eq_define_F_1}.  For the similar statement with $T^\og$ replacing
by $T^{\rm Otype}$, we let $\cT$ to be those 
ordered $B$-subgraphs of types $T^\og$ and edge lengths $\mec k\ge\bec\xi$.
\end{proof}

\subsection{Remarks on the Length-Multiplicity Formula}

Lemma~\ref{le_length_mult}
has the effect of writing a ``decoupled'' expression
for \eqref{eq_certified_type}
in terms of edge lengths, $\mec k$, of walks and their
multiplicities $\mec m$, an expression which we will call a
{\em certified dot convolution} (in Definition~\ref{de_cert_dot_conv}).
This is our approach to proving Theorem~\ref{th_main_certified_walks}.

We will apply Lemma~\ref{le_length_mult} to $B$-types
$T^{\rm Otype}$ where $\cC_n(B)$ is algebraic, i.e.,
where there are polynomials $p_0,p_1,\ldots$ such that
any $S_\Bg^\og$ of this type has
$$
\EE_{G\in\cC_n(B)}\Bigl[ \#\bigl([S_\Bg^\og]\cap G\bigr) \Bigr]
=
c_0(S_\Bg) + c_1(S_\Bg)/n + \cdots + c_{r-1}(S_\Bg)/n^{r-1}
O(1)g(\#E_S)/n^r,
$$
where $c_i(S_\Bg)=p_i(\mec a_S)$, and $g$ is a function of
growth $1$ depending on $r$ (but not $S$).
The main technical difficulty in proving 
Theorem~\ref{th_main_certified_walks} is to
sum such $p_i(\mec a_S)$ over all $S$ of a given $B$-type,
i.e., to sum $p_i(\mec a_W)$ over all wordings, $W$, of type
$T^{\rm Otype}$;
this can be done by adapting the methods of 
\cite{friedman_random_graphs}.

\subsection{The Methods of Broder-Shamir and Multiplicities Greater
Than One}
\label{su_methods_broder_shamir}

To understand how we use Lemma~\ref{le_length_mult}, it is helpful
to keep in mind the special case where $T^\og$ is the homotopy type
of a cycle.
In this case $T$ has one vertex and one edge, a whole-loop.  It
follows that the formula $\mec k\cdot\mec m=k$ reads
$k_1m_1=k$, and hence $k_1$ must be a divisor of $k$.
It then follows that the case $m_1=1$, i.e., $k_1=k$, gives the
dominant terms in the asymptotic expansions (see Section~6 of Article~I),
and when $m_1\ge 2$ we have $k_1\le k/2$, which gives terms
whose coefficients are functions of growth $(d-1)^{1/2}$ or smaller.
This calculation goes back to \cite{broder}.

Even for general $T^\og$, of any order, these observations guide
our style of estimating \eqref{eq_certified_type}.
Namely, if $\mec m\ge \mec 2$, i.e.,
$m(e)\ge 2$ for each $e\in E_T$, then such a term gives rise
to expansion coefficients of growth $(d-1)^{1/2}$ or smaller.
So our approach to estimating 
\eqref{eq_certified_type} will separate the case $\mec m\ge\mec 2$
from all the other cases (see Subsection~\ref{su_growth_lemma}).
Moreover, for any $E'\subset E_T$ we
will consider the case where $m(e)=1$ for $e\in E'$ and
$m(e)\ge 2$ for $e\notin E'$.
Intuitively speaking, in this case the $\mec k\cdot \mec m=k$
condition gives larger terms (and expansion coefficients)
when the part of $\mec k\cdot\mec m$ involving the $e\in E'$
is as large as possible (i.e., close to $k$), and the part
involving $e\notin E'$ is as small as possible.
What we will show (similarly to \cite{broder,friedman_random_graphs}),
roughly speaking,
is that for a fixed $E'$, when each $e\notin E'$
has values $k(e)=1,2,3,\ldots$, we get coefficients that decay
geometrically in these values of $k(e)$ for $e\notin E'$.
We make this precise and more general in
Lemma~\ref{le_third_main_lemma}.

\section{Outline of the Proof of Theorem~\ref{th_main_certified_walks}}
\label{se_proof_outline}

In this section we outline the proof of Theorem~\ref{th_main_certified_walks}
and introduce some notation and concepts that we will need.
The proof of Theorem~\ref{th_main_certified_pairs} is similar although
involves more cumbersome notation.

Our proof is based on Lemma~\ref{le_length_mult}; we require some 
definitions.

\subsection{Definitions: The Dot Convolution and Polyexponential Functions}

\begin{definition}\label{de_cert_dot_conv}
Let $f,g$ be two functions $\naturals^s\to \complex$ for some integer
$n\ge 1$, and fix a $\bec\cert\in\naturals^s$.  We defined the
{\em $\bec\cert$-certified dot convolution of $f$ and $g$} to be the
function $\naturals\to\complex$, denoted $f\star_{\ge\bec\cert} g$,
given by
\begin{equation}\label{eq_cert_dot_conv_def}
(f\star_{\ge\bec\cert} g)(k) \eqdef
\sum_{\mec k\cdot\mec m=k,\ \mec k\ge \bec\cert}
f(\mec k) g(\mec m) .
\end{equation}
\end{definition}

\begin{definition}
A function 
$f=f(\mec k)=f(k_1,\ldots,k_m)\from \naturals^m\to\complex$ 
is called a {\em polyexponential
function} if it is a linear combination of products of 
(univariate) polyexponential functions 
(see Subsection~\ref{su_asymptotic_expansions})
in the individual variables $k_1,\ldots,k_m$.  Equivalently, 
$f$ is polyexponential
if it can be written as 
$$
f(\mec k) = \sum_{\bec\beta\in M\subset\complex^m}
\bec\beta^{\mec k} p_{\bec\beta}(\mec k)
$$
where 
(1) the sum is over $\bec\beta$ in some finite subset, $M$, of $\complex^m$,
(2) $p_{\bec\beta}=p_{\bec\beta}(\mec k)$ is a (multivariate) polynomial, and
(3) we use the ``tensor'' notation
$$
\bec\beta^{\mec k} = \beta_1^{k_1}\ldots\beta_m^{k_m},
$$
and (4) we adopt the convention
that for any $\beta_i=0$, the term $\beta_i^{k_i}$ is omitted, and
$p_{\bec\beta}(\mec k)$ is nonzero for only finitely many values of $k_i$,
and for all fixed values of all such $k_i$, the function
$p_{\bec\beta}(\mec k)$ is a polynomial in the remaining $k_1,\ldots,k_m$
(and this polynomial can depend on the values of $k_i$ with $\beta_i=0$).
We refer to all the $\beta_i$ appearing in the finite set of $\bec\beta$
as the {\em bases} of $f$.
\end{definition}
Our convention for $\beta=0$ above is motivated in Article~I; the point
is that if $J$ is a Jordan block corresponding to an eigenvalue $0$,
then $J^k$ is nonzero for only finitely many values of $k$.

\begin{definition}
Let $\rho\ge 0$ be a real number and $m\in\naturals$.
A function $f\from\naturals^m\to\complex$ is called a {\em of growth rate 
$\rho$} 
(or simply {\em growth $\rho$}) if for every $\epsilon>0$ we have
$$
\bigl| f(\mec k) \bigr| \le (\rho+\epsilon)^{\mec k\cdot\mec 1} 
= 
(\rho+\epsilon)^{k_1+\cdots+k_m} 
$$
for $\mec k\cdot\mec 1$ sufficiently large; equivalently for every
$\epsilon>0$ there is a $C=C(\epsilon)\in\reals$ such that
\begin{equation}\label{eq_multi_growth_definition}
\bigl| f(\mec k) \bigr| 
\le C(\epsilon)(\rho+\epsilon)^{\mec k\cdot\mec 1} 
\end{equation} 
for all $\mec k$.
\end{definition}
Often \eqref{eq_multi_growth_definition} is more convenient for estimates,
even though the first inequality is simpler.

The following definition generalizes the notion of a $(B,\nu)$-bounded
and $(B,\nu)$-Ramanujan functions
univariate function to multivariate functions.
\begin{definition}
Let $B$ be a graph and $\nu>0$ be real numbers.
We say that a function $f\from\naturals^m\to\complex$ is 
{\em $(B,\nu)$-bounded} if there is a polyexponential function
$P\from\naturals^m\to\complex$ whose bases are bounded in absolute
value by $\mu_1(B)$ such that
$f-P$ is of growth rate $\nu$; furthermore we say that $f$ is
$(B,\nu)$-Ramanujan if the bases are a subset of the $\mu_i(B)$, 
i.e., of the eigenvalues of $H_B$.
\end{definition}

\subsection{Outline of the Proof of Theorem~\ref{th_main_certified_walks}}

Our proof of 
Theorem~\ref{th_main_certified_walks} is based on the following three
lemmas.
We prove the first lemma immediately, since the proof is very easy.

\begin{lemma}\label{le_first_main_lemma}
Let $T^\og$ be any ordered graph, and $\bec\cert\from E_T\to\naturals$.
Let
$$
\omega(M) \eqdef
\sum_{\bec\cert\cdot \mec m= M} F_2(\mec m)
$$
where $F_2(\mec m)$ is as in \eqref{eq_define_F_1}, i.e.,
$$
F_2(\mec m) \eqdef \legal(T^\og,\mec m) \ .
$$
Then $\omega(M)$ is of growth rate
$$
\mu_1\bigl( \VLG(T,\bec\cert) \bigr) \ .
$$
\end{lemma}
\begin{proof}
Let $T'=\VLG(T,\bec\cert)$.  
The SNBC closed walks in $T^\og$ 
inject into the set of SNBC closed walks in $T'$, in a way that
if the walk in $T$ has multiplicities $\mec m$, then its length
in $T'$ is 
$\bec\xi\cdot\mec m$.  It follows that
$\omega(M) \le \Trace(H_{T'}^M)$ which is bounded by
$\#\Edir_{T'}$ times $\mu_1(T')^M$.
\end{proof}

\begin{lemma}\label{le_second_main_lemma}
Let $B$ be a graph, $T^{\rm Otype}$ a $B$-type, and
$\{\cC_n(B)\}_{n\in N}$ an algebraic model.
Let $F_1(\mec k)$ be as in \eqref{eq_define_F_1}, i.e.,
$$
F_1(\mec k,n) \eqdef
\EE_{G\in\cC_n(B)}\Bigl[\subgr_B\bigl(T^{\rm Otype},\mec k;G_\Bg)\Bigr]\ .
$$
Then for any $r$ there is
a constant $C$ such that for all $\mec k\cdot\mec 1\le n^{1/2}/C$
(and $n\in N$)
\begin{equation}\label{eq_Otype_lemma}
F_1(\mec k,n)=c_0(\mec k)+\cdots+c_{r-1}(\mec k)/n^{r-1}+
O(1/n^r)C_r(\mec k),
\end{equation} 
whose coefficients, $c_i(\mec k)$ are
polyexponential functions
of $\mec k$ (i.e., purely polyexponential, with no error term), whose
bases are eigenvalues of the type,
and where $C_r(\mec k)$ is a function
of growth $\mu_1(B)$.
\end{lemma}

Our third lemma is a general fact about dot convolutions.

\begin{lemma}[The Certified Dot Convolution Lemma]
\label{le_third_main_lemma}
Let $f,g$ be two functions $\naturals^s\to \complex$ for some integer
$n\ge 1$, and fix a $\bec\cert\in\naturals^s$.  Assume that
(1) $f$ is a polyexponential function whose bases are bounded in
absolute value by some real $\beta>0$, and that (2) the function
(of $M$)
$$
\omega(M) \eqdef
\max_{\bec\cert\cdot \mec m= M} |g(\mec m) |
$$
is of growth rate $\rho$.
Then the $\bec\cert$-certified dot convolution of $f$ and $g$, i.e.,
$$
(f\star_{\ge\bec\cert} g)(k) \eqdef
\sum_{\substack{\mec k\cdot\mec m=k,\ \mec k\ge \bec\cert\\
\mec k,\mec m\in\naturals^s}}
f(\mec k) g(\mec m)
$$
is an approximate polyexponential function whose
bases are those of $f$ and whose
error term is of growth
$$
\max \bigl( \beta^{1/2},\rho,1 \bigr).
$$
\end{lemma}
The $\beta^{1/2}$ in the above equation is due to the contribution
of $\mec m\ge 2$ (see Subsection~\ref{su_methods_broder_shamir}
for its appearance, that essentially goes back to \cite{broder}).

The proofs of each of these lemmas are simple adaptations of the
proofs of the special cases given in \cite{friedman_random_graphs},
and the proofs are independent of one another.
Lemma~\ref{le_third_main_lemma} will be proven in the next section.

When we generalize our proof of 
Theorem~\ref{th_main_certified_walks} to
Theorem~\ref{th_main_certified_pairs}, 
Lemmas~\ref{le_first_main_lemma} and \ref{le_third_main_lemma} will
be used verbatim, but
Lemma~\ref{le_second_main_lemma} will require modification.
To make this modification easy, in the next
subsection we will state a precursor to
Lemma~\ref{le_second_main_lemma} that can be used to count the expected
number of SNBC walks of homotopy type $T^\og$ certified by $\bec\xi$,
and similarly for {\em certified pairs}
(defined in Section~\ref{se_pairs_prelim}).

Roughly speaking, the proof of Theorem~\ref{th_main_certified_walks} 
will apply Lemma~\ref{le_third_main_lemma} to each coefficient
$c_i(\mec k)$ and to $C_r(\mec k)$ of
\eqref{eq_Otype_lemma} in Lemma~\ref{le_second_main_lemma}.

\subsection{A Lemma on Types or Regular Languages}

The following lemma easily implies Lemma~\ref{le_second_main_lemma}
and the generalization of this lemma needed to count the expected number
of certified pairs.
This lemma is also stated so that we can easily use it when we 
generalize our proof of Theorem~\ref{th_main_certified_walks}
to that of Theorem~\ref{th_main_certified_pairs}.

If $B,T$ are graphs and $W$ is a $B$-wording of $T$, then we write
$\mec a_W$ for $\mec a_{S_\Bg}$ where $S_\Bg=\VLG(T,W)$; in other words
$\mec a_W\from E_B\to\integers_{\ge 0}$ is the function where
$\mec a_W(e)$ is the number of times $e_B$ occurs in either
orientation in $W$ restricted to any orientation of $T$.

\begin{lemma}\label{le_precursor}
Let $B$ be a graph and
$T^{\rm Otype}=(T^\og,\cR)$ an
ordered $B$-type.  Say that for each wording $W$
of type $T^{\rm Otype}$
we are given a function
$f_W=f_W(n)\from N\to\complex$ for an infinite set $N\subset \naturals$.
Consider any
$r\in\naturals$, $C>0$, and function
$g$ of growth $1$ for which
$$
f_W(n) = c_0(W)+ \cdots + c_{r-1}(W) /n^{r-1} + O(1/n^r)g(|W|)
$$
where $|W|$ is the sum of the edge-lengths of $W$,
for a universal constant in the $O(1/n^r)$
in the range $1\le |W|\le n^{1/2}/C$.  Furthermore assume that
there are polynomials $p_0,\ldots,p_{r-1}$ such that
$c_i(W)=p_i(\mec a_W)$ for all $i=0,\ldots,r-1$.  Then
$$
f(\mec k,n) \eqdef \sum_{W\in T^{\rm Otype}[\mec k]} f_W(n)
$$
has an asymptotic expansion
$$
f(\mec k,n)=
c_0(\mec k)+\cdots+c_{r-1}(\mec k)/n^{r-1}+
O(1/n^r)C_r(\mec k)
$$
in the range $\mec k\cdot\mec 1\le n^{1/2}/C$ (and for $n\in N$),
whose coefficients, $c_i(\mec k)$ are
polyexponential functions
of $\mec k$ (i.e., purely polyexponential, with no error term), whose
bases are a subset of the eigenvalues of $T^{\rm Otype}$
(i.e., of the regular languages $\cR(e)$ with $e$ varying over $\Edir_T$)
and where $C_r(\mec k)$ is a function
of growth $\mu_1(B)$.
Furthermore, if the $p_i$ vanish for $i<i_0$ for some $i_0$,
then the coefficients $c_i$ vanish for $i<i_0$.
\end{lemma}
Note that this lemma is stated purely in terms of one
ordered $B$-type, and can
be viewed as a lemma on $\cR$, which is essentially a cartesian product of
regular languages.
Hence this lemma is merely a lemma regarding functions on
products of regular languages (separate from any applications).
Of course, in applications the $i_0$ above will be the order $T$,
as the $f_W(n)$ expansion comes from an algebraic model.

\section{Proof of The Certified Dot Convolution Lemma}
\label{se_cert_dot}

In this Section we prove the
Certified Dot Convolution
Lemma, i.e., Lemma~\ref{le_third_main_lemma}.
The case where $\bec\cert=1$ (i.e., the dot convolution)
is the case addressed in 
\cite{friedman_random_graphs}.  It suffices to adapt the proof there
to this more general situation, and check that everything still works.
We shall break the proof into a number of subsections, and complete
the proof in the last subsection.

\subsection{First Step in the Proof: The Growth Lemma}
\label{su_growth_lemma}

The main idea in the proof is to divide the sum in the certified
dot convolution by fixing some of the $m_1,\ldots,m_s$ to be $1$, and
the others to be at least $2$ (there are only $2^s$ ways to do this,
and $2^s$---for fixed $s$---is a constant).  We first consider part
of this sum in the case that
$m_i\ge 2$ for all $i$, and we prove a lemma about the growth rate
of this part of the sum.
The situation where some of the $m_i$ are fixed to
be $1$ will also use this growth lemma (applied to those $i$ for which
$m_i\ge 2$).

\begin{lemma}\label{le_growth_lemma}
Let $f,g$ be functions from $\naturals^s\to\complex$,
for some integer $s\ge 1$, and let $\bec\cert\in\naturals^s$;
also assume that $f$ is of growth $\beta$, 
and the function
$$
\omega(M) \eqdef
\max_{\bec\cert\cdot \mec m= M} |g(\mec m)|
$$
is (as a function of $M$) of growth rate $\rho$.
Consider
\begin{equation}\label{eq_restricted_dot_conv}
P(k)\eqdef
\sum_{\mec k\cdot\mec m = k,\ \mec k\ge\bec\cert,\ \mec m\ge \mec 2 } 
f(\mec k) g(\mec m),
\end{equation}
i.e., the certified dot convolution summed only over those $\mec m\ge \mec 2$
(where $\mec 2$ is the vector whose components are all $2$).
Then $P(k)$ is of growth 
$$
\max(\beta^{1/2},\rho,1).
$$
\end{lemma}
We remark that the $1$ in the max function above is needed
in passing from
\eqref{eq_K_plus_M_in_exponent} to
\eqref{eq_K_plus_M_in_exponent_two}.
\begin{proof}
For any $\mec k\ge \bec\cert$ and $\mec m\ge \mec 2$
we have
$$
0 \le (k_i-\cert_i)(m_i-2)
$$
for any $i=1,\ldots,n$,
and hence
\begin{equation}\label{eq_ckm_ineq}
2(k_i-\cert_i)+\cert_i m_i\le k_im_i.
\end{equation}
So if we set 
$$
K(\mec k)  \eqdef \sum_i 2(k_i-\cert_i) , \quad
M(\mec m)\eqdef \bec\cert\cdot\mec m=\sum_i \cert_im_i,
$$
then summing over $i$ in \eqref{eq_ckm_ineq} yields
\begin{equation}\label{eq_K_M_sum_inequality}
K(\mec k) +M(\mec m)\le \sum_i k_im_i = k \ .
\end{equation} 
Keeping in mind that $\bec\cert$ is fixed, we see that for all
$\mec k,\mec m$ such that
\begin{equation}\label{eq_eligible_mc}
\mec k\cdot\mec m=k,\quad \mec k\ge \bec\cert, \quad \mec m\ge \mec 2
\end{equation}
all hold, we have that 
for any $\epsilon>0$:
\begin{enumerate}
\item we have $k_1+\cdots+k_s=K(\mec k)/2 + (\cert_1+\ldots+\cert_s)$;
\item
hence
\begin{align*}
|f(\mec k)| & \le C(\epsilon) (\beta+\epsilon)^{k_1+\cdots+k_s}
\\
& = C(\epsilon) (\beta+\epsilon)^{K(\mec k)/2 + (\cert_1+\ldots+\cert_s)} = 
C_1(\epsilon)(\beta+\epsilon)^{K(\mec k)/2}
\end{align*}
for a constant $C_1(\epsilon)=C(\epsilon)(\beta+\epsilon)^{(\cert_1+\cdots+
\cert_s)/2}$ (depending on $\epsilon,f,\beta,\bec\cert$);
\item the assumption on $\omega(M)$ implies that 
$$
|g(\mec m)| \le C_2(\epsilon)(\rho+\epsilon)^{M(\mec m)}
$$
for some constant $C_2=C_2(\epsilon)$;
\item 
and therefore
\begin{equation}\label{eq_K_plus_M_in_exponent}
|f(\mec k)|\,|g(\mec m)| 
\le C_3(\epsilon) \max\bigl( (\beta+\epsilon)^{1/2},\rho+\epsilon 
\bigr)^{K(\mec k)/2 + M(\mec m)};
\end{equation} 
in view of \eqref{eq_K_M_sum_inequality} 
we have
\begin{equation}\label{eq_K_plus_M_in_exponent_two}
|f(\mec k)|\,|g(\mec m)|
\le C_3(\epsilon) \max\bigl( (\beta+\epsilon)^{1/2},\rho+\epsilon,1 \bigr)^k
\end{equation} 
for a constant $C_3=C_3(\epsilon)$.
\end{enumerate}
But there are (crudely) at most $k^{2s}$ pairs $\mec k,\bec\cert$ satisfying
\eqref{eq_eligible_mc}, since each (of the $2s$ variables)
$k_i$ and $m_i$ is some integer
between $1$ and $k$.
Hence for any $\epsilon>0$ we have
$$
|P(k)| \le C_3(\epsilon) k^{2s} \max\bigl( (\beta+\epsilon)^{1/2},
\rho+\epsilon \bigr)^k,
$$
and therefore $P(k)$ is of growth rate at most
$\max( \beta^{1/2},\rho,1 )$.
\end{proof}

\subsection{Generating Functions and Convolutions}

We will be interested in how
polyexponentials and functions of bounded growth
behave under convolution.
We shall replace some of the proofs of \cite{friedman_random_graphs}
with simpler proofs based on generating functions.
The only mildly awkward point is that we are interested in
functions $\naturals\to\complex$, instead of the more standard
functions $\integers_{\ge 0}\to\complex$.
Let us first deal with the latter functions.

\begin{definition}\label{de_additive_convolution}
Let $P,Q$ be functions $\integers_{\ge 0}\to\complex$.
We define the {\em additive convolution of $P,Q$} to be
$$
(P \ast Q)(k) = \sum_{i=0}^k P(i) Q(k-i).
$$
If
$Q\from\integers_{\ge 0}\to \complex$ is any function, we define a formal
power series, called the {\em generating function of $Q$} via
$$
{\rm Gen}_Q(z) \eqdef \sum_{i=0}^\infty Q(i)z^i .
$$
By the {\em degree} of a non-zero polynomial $p=p(z)$ we mean the degree
of its leading term (so $\deg p=0$ if $p$ is a constant, and $\deg p$ is
undefined or $-\infty$
if $p(z)=0$ is the zero polynomial).
If $f\from\naturals\to\complex$, we define its {\em generating function}
to be that of $f$ where we view $f$ as a function on $\integers_{\ge 0}$
by setting $f(0)=0$.
\end{definition}
The following lemma is straightforward.
\begin{lemma}\label{le_additive_convolution}
Let $P,Q$ be functions $\integers_{\ge 0}\to\complex$.  Then
\begin{enumerate}
\item For any real $\rho>0$,
${\rm Gen}_P(z)$ is holomorphic in $|z|<1/\rho$
iff for any $\epsilon>0$, $|P(i)|<(\rho+\epsilon)^i$ for sufficiently
large $i$;
\item
\begin{equation}\label{eq_additive_conv}
{\rm Gen}_{P \ast Q}(z)
 = {\rm Gen}_P (z)
 {\rm Gen}_Q (z)  ;
\end{equation}
\item
for any nonzero $\beta\in\complex$ and any $d\in\integers_{\ge 0}$,
$P$ is of the form $P(j)=q(j)\beta^j$ for a polynomial $q$ of degree
$d$ iff
${\rm Gen}_P(z)=p(z)/(1-\beta z)^{d+1}$ in a neighbourhood of $z=0$
for a polynomial, $p$, of degree
at most $d$ with $p(1/\beta)\ne 0$.
\item
for any set of nonzero complex numbers $R$,
a nonzero function $P=P(j)$ is a sum of functions of the form $q(j)\beta^j$
with $q$ a polynomial and $\beta\in R$ iff
$$
{\rm Gen}_P(z)=p_1(z)/p_2(z)
$$
where $p_1,p_2$ are polynomials such that $\deg p_1<\deg p_2$ and
all roots of $p_2$ are reciprocals of elements of $R$.
\item
for any set of nonzero complex numbers $R$,
a nonzero function $P=P(j)$ is a sum of functions of the form $q(j)\beta^j$
with $q$ a polynomial and $\beta\in R$ iff the formal power series
${\rm Gen}_P(z)$ is a meromorphic function in the complex plane that
tends to zero as $z$ tends to infinity and whose poles are
reciprocals of $R$.
\end{enumerate}
\end{lemma}
\begin{proof}
Item~(1) follows from complex analysis, and Item~(2) is immediate.
Item~(3) follows from induction by differentiating $k$ times in $z$
the relation
$$
\sum_{j=0}^\infty (\beta z)^j = 1/(1-\beta z).
$$
Item~(4) follows from (3), and the standard
``partial fractions'' fact that
a rational function
$q(z)/p(z)$ where $p,q$ are polynomials with $\deg q<\deg p$ can be
written as a sum of $q_i(z)/(r_i-z)^{d_i}$ where each $q_i$ is a polynomial
of degree less than $d_i$ and the $r_i$ are the roots of $p$.
Item~(5) is a restatement of Item~(4).
\end{proof}

If $P\from\naturals\to\complex$, by the above definition we have
$$
{\rm Gen}_P(z)=\sum_{j=1}^\infty P(j)z^j = z \sum_{j=0}^\infty z^j P(j+1).
$$

\begin{corollary}
For a nonzero $P\from\naturals\to\complex$ and
$\beta_1,\ldots,\beta_k\in\complex\setminus\{0\}$ we have that
\begin{enumerate}
\item
$P$ is a polyexponential with bases $\beta_1,\ldots,\beta_k$
iff the power series ${\rm Gen}_P(z)$ defines
a meromorphic function in the complex plane
that vanishes at
$z=0$, that is bounded as $z$ tends to infinity, and whose poles are
reciprocals; and
\item
for any $\rho>0$ smaller than all $1/|\beta_i|$,
$P$ is a approximate
polyexponential with bases $\beta_1,\ldots,\beta_k$ and error
growth $\rho$
iff the power series ${\rm Gen}_P(z)$ defines
a meromorphic function in $|z|<1/\rho$ that vanishes at $z=0$
and each of whose poles
is of the form $1/\beta_j$ for some $j$.
\end{enumerate}
\end{corollary}

\begin{corollary}\label{co_additive_convolution}
The set of polyexponentials whose bases lie in a set
$R\subset \complex$ is closed under convolution.
Similarly for the set of approximate polyexponentials
whose bases lie in a set $R\subset \complex$ and
error growth is $\rho$.
\end{corollary}
\begin{proof}
This follows from \eqref{eq_additive_conv}
and the above corollary.
\end{proof}

[One can compare the above proof with the proof
of Lemma~2.14 of \cite{friedman_random_graphs} in the case
$m_1=\cdots=m_t=1$;
in this article we use generating functions to replace the
various individual computations and sublemmas in
\cite{friedman_random_graphs}.]

\subsection{Polyexponential Convolved with
a Bounded Function}

The last main result we need to prove Lemma~\ref{le_third_main_lemma}
involves convolving a polyexponential function with one of bounded
growth.

\begin{lemma}\label{le_poly_with_bounded}
Let $P,Q$ be functions $\integers_{\ge 0}\to\complex$, where
$P$ is a polyexponential function with bases in $R\subset \complex$,
and $Q$ is of growth $\rho$.
Then
$P\ast Q$ is an approximate polyexponential function, with bases in
$R$ and error base $\rho$.
\end{lemma}
\begin{proof}
We have that ${\rm Gen}_P(z)$ is a rational function with poles whose
reciprocals are in $R$, and ${\rm Gen}_Q(z)$ is holomorphic in
$|z|< 1/\rho$.  It follows that for all $|z|<1/\rho$,
$$
{\rm Gen}_{P\ast Q}(z)=
{\rm Gen}_P(z) {\rm Gen}_Q(z)
$$
is meromorphic with poles at the reciprocals of $R$.
\end{proof}

\subsection{Proof of Lemma~\ref{le_third_main_lemma}}

\begin{proof}[Proof of Lemma~\ref{le_third_main_lemma}.]
For each $\mec m\in\naturals^s$, the set of $i$ for which
$m_i=1$ determines a unique subset, $\cN=\cN(\mec m)$, of $\{1,\ldots,n\}$.
So for a fixed subset, $N$, of $\{1,\ldots,n\}$, let
$(f \star_{\ge \bec\cert,N} g)(k)$ denote the sum in 
\eqref{eq_cert_dot_conv_def}
restricted to those pairs $\mec k,\mec m$ for which $\cN(\mec m)=N$.
Clearly
\begin{equation}\label{eq_sum_over_N}
f \star_{\ge\bec\cert} g = \sum_{N\subset \{1,\ldots,n\}} f \star_{\ge\bec\cert,N} g,
\end{equation}
and so it suffices to prove the theorem with $f\star_{\ge\bec\cert} g$
replaced with $f\star_{\ge\bec\cert,N} g$.

Without loss of generality we may assume $N=\{1,\ldots,s\}$.
By the bilinearity of $f\star_{\mec\ge c,N} g$ in $f$, it suffices to
to prove the theorem for $f(\mec k)$ of the form
$$
f(\mec k) = f_1(k_1)f_2(k_2)\ldots f_s(k_s),
$$
where the $f_i$ are polyexponential functions with bases bounded
above in absolute value by $\beta$.
Setting
$$
\mec k'= (k_1,\ldots,k_{s'}),\quad
\mec k'' = (k_{s'+1},\ldots,k_s), \quad
\mec m'' = (m_{s'+1},\ldots,m_s),
$$
$$
\bec\cert' = (\cert_1,\ldots,\cert_{s'}),\quad 
\bec\cert'' = (\cert_{s'+1},\ldots,\cert_s),
$$
we have
$$
(f\star_{\ge\bec\cert,N} g)(k) =
\sum_{\mec k'\cdot\mec 1+\mec k''\cdot\mec m''=k,\ \mec m''\ge \mec 2,\ 
\mec k\ge\bec\cert}
f_1(k_1)f_2(k_2)\ldots f_s(k_s) g(\mec 1,\mec m''),
$$
where $\mec 1$ is the vector of $1$'s
(i.e., the vector $(m_1,\ldots,m_{s'})$, since these $m_i$ are fixed to 
to equal $1$).
In the above sum we can set
$$
K_1 = \mec k\cdot\mec 1,\quad K_2=\mec k''\cdot\mec m''
$$
and write
$$
(f\star_{\ge\bec\cert,N} g)(k)
=
\sum_{K_1+K_2=k} F_1(K_1) F_2(K_2) ,
$$
where
$$
F_1(K_1) = \sum_{\substack{\mec k'\cdot\mec 1=K_1 \\ \mec k'\ge \bec\cert'}}
f_1(k_1)\ldots f_{s'}(k_{s'}),
$$
and
$$
F_2(K_2) =
\sum_{\substack{ \mec k''\cdot\mec m''=K_2\\ \mec m''\ge \mec 2,\ \mec k''
\ge \bec\cert''}}
f_{s'+1}(k_{s'+1})\ldots f_s(k_s)g(\mec 1,\mec m'') .
$$
It follows that:
\begin{enumerate}
\item $f\star_{\ge\bec\cert,N} g = F_1 \ast F_2$;
\item By Lemma~\ref{le_growth_lemma}, 
$F_2$ is of growth bounded by
$$
\max\bigl( \beta^{1/2},\rho \bigr).
$$
\item $F_1$ is polyexponential with bases bounded 
(above in absolute value) by $\beta$ by the
following reasoning:
setting $\mec\ell=\mec k'-\bec\cert'+\mec 1$ we have
$$
F_1(K_1) = \sum_{\substack{\mec k'\cdot\mec 1=K_1 \\ \mec k'\ge \bec\cert'}}
f_1(k_1)\ldots f_{s'}(k_{s'})
$$
$$
= \sum_{\mec\ell\in\naturals^s} f_1(\ell_1-\cert_1+1)\ldots f_{s'}(\ell_{s'}-\cert_{s'}+1)
= (\widetilde f_1\ast\cdots\ast \widetilde f_{s'})(K_1-C),
$$
where $C=(\cert_1-1)+\cdots+(\cert_{s'}-1)$, and where
$\widetilde f_i(k_i)\eqdef f_i(k_i-\cert_i+1)$;
clearly the $\widetilde f_i$ are also polyexponential with bases
bounded by $\beta$, and so 
Lemma~\ref{le_additive_convolution} implies that 
$\widetilde f_1\ast\cdots\ast \widetilde f_{s'}$ is as well, and hence so
is the function of $K_1$ given by
$(\widetilde f_1\ast\cdots\ast \widetilde f_{s'})(K_1-C)$.
\end{enumerate}
Hence
Lemma~\ref{le_poly_with_bounded} implies that
$f\star_{\ge\bec\cert,N} g = F_1 \ast F_2$ is an approximate polyexponential
function with bases bounded by $\beta$ and error growth
$\max\bigl( \beta^{1/2},\rho,1 \bigr)$.
Summing over all $N\subset\{1,\ldots,n\}$, as in \eqref{eq_sum_over_N},
yields 
Lemma~\ref{le_third_main_lemma}.
\end{proof}

\subsection{Improvement to Lemma~\ref{le_third_main_lemma}}

We remark that the $\beta^{1/2}$ in
Lemma~\ref{le_third_main_lemma} can be improved to
$\beta^{1/N}$ for any positive integer, $N$, although
$(f\star_{\ge \bec\cert} g)(k)$ would conceivably need
to be a different polyexponential for
each congruence class of $k$ modulo the least common multiple
of $1,\ldots,N-1$;
see \cite{friedman_kohler}.  For example, the expected number of
closed, strictly non-backtracking walks of length $k$ in a random
$d$-regular graph whose trajectory is a
``simple loop'' is known  (\cite{friedman_alon}, Theorem~5.2) to equal
$$
p_0(k) + O(1/n) p_1(k),
$$
where
$$
p_0(k) = \sum_{k'|k} (d-1)^{k'} ,
$$
a sum over all $k'$ dividing $k$.  
So, for example, up to an error of growth
rate $(d-1)^{1/3}$, for $k$ even there is a term of order
$(d-1)^{k/2}$, and for $k$ odd there is no such term.
Similarly, up to an error of growth rate $(d-1)^{1/4}$, there are
terms of order $(d-1)^{k/2}$ and $(d-1)^{k/3}$ depending on the 
value of $k$ modulo $6$.
More generally, for the ``simple loop'' in the case where $B$ is $d$-regular,
the polyexponentials that appear up
to error $(d-1)^{1/N}$ clearly depend on which of $2,\ldots,N-1$
divide $k$.
In this series of papers we will need only the case $N=2$, 
i.e., that of Lemma~\ref{le_third_main_lemma},
where
the polyexponential in $k$ is the same for all integers $k$.

\section{Lemma on Regular Languages}
\label{se_regular}

The goal of this section is to generalize
Lemma~2.11 of \cite{friedman_random_graphs} (also quoted as
Lemma~5.10 in \cite{friedman_alon})---which is
our Corollary~\ref{co_regular_unweighted} below---needed to
prove Lemma~\ref{le_second_main_lemma}; the proof we give here
is much simpler and shorter than that in \cite{friedman_random_graphs}.

\begin{definition}
Let $\cA$ be an alphabet (i.e., a finite set).  For each
word $w\in\cA^*$ (i.e., a finite sequence of elements of $\cA$)
let ${\rm occur}(w)$ be the vector $\cA\to\integers_{\ge 0}$
taking $\sigma\in\cA$ to the number of occurrences of $\sigma$ in $w$.
If $\bec\beta\from \cA\to\complex$ is any function, we use the tensor
notation
$$
\bec\beta^{{\rm occur}(w)} \eqdef 
\prod_{\sigma\in\cA} \beta(\sigma)^{{\rm occur}(w)(\sigma)}
$$
\end{definition}

\begin{lemma}\label{le_regular}
Let $L$ be a regular language over an alphabet $\cA$, and $\cD$ a
finite automaton accepting $L$.  For each $k\in\integers_{\ge 0}$,
let $L_k=L\cap \cA^k$, i.e., the number of words in $L$ of length $k$.
For any $\bec\beta\from\cA\to\complex$, we have that
$$
f(k) \eqdef \sum_{w\in L_k} P(w), \quad\mbox{where}\quad
P(w) \eqdef \bec\beta^{{\rm occur}(w)}
$$
is a polyexponential whose bases are some subset of the eigenvalues
of the (square) matrix, $M$, indexed on the states of $\cD$, where
$M_{s_1,s_2}$ is the sum of $\beta(\sigma)$ over all $\sigma\in\cA$
which are transitions from $s_1$ to $s_2$.
\end{lemma}
\begin{proof}
Clearly for any $k$ we have that the $s_1,s_2$ entry of $M^k$ is the 
sum of $P(w)$ of all words of length $k$, $w$, taking $s_1$ to $s_2$.
Hence $L_k$ is the sum of entries $(M^k)_{s_1,s_2}$ where $s_1$ is 
the initial state of $\cD$, and $s_2$ is an accepting (i.e., final)
state of $\cD$.
By Jordan canonical form, each function $g(k)=(M^k)_{s_1,s_2}$ with
$s_1,s_2$ fixed is a polyexponential in the eigenvalues of $M$.
\end{proof}

\begin{corollary}\label{co_regular}
The conclusion of Lemma~\ref{le_regular} holds with $P$ replaced by
$$
P(w) = \prod_{\sigma} g_\sigma\bigl({\rm occur}(w)(\sigma)\bigr)
$$
where each $g_\sigma$ is a function of the form
$$
g_\sigma(x_\sigma) = 
x_\sigma(x_\sigma-1)\ldots (x_\sigma-\ell_\sigma+1) 
\beta_\sigma^{x_\sigma-\ell_\sigma}
$$
for $\beta_\sigma\in \complex$ and $\ell_\sigma\in\integers_{\ge 0}$.
\end{corollary}
\begin{proof}
Lemma~\ref{le_regular} proves this in the case $\ell_\sigma=0$ for
all $\sigma$.
Now partially differentiate this function $P(w)$
$\ell_\sigma$ times in the variable $x_\sigma$ for each $\sigma\in \cA$.
\end{proof}

The case where $\bec\beta=(1,\ldots,1)$ yields the following corollary.

\begin{corollary}\label{co_regular_unweighted}
The conclusion of Lemma~\ref{le_regular} holds for any function
$P(w)$ which is a polynomial in variables ${\rm occur}(w)$,
and in this case the bases of $f(k)$ are the eigenvalues of the
regular language.
\end{corollary}

We remark that some natural variants of the functions $f(k)$ in
the above lemma and corollaries are not polyexponential functions.
For example, if $f\from\integers_{\ge 0}^\cA\to\integers$
is the multivariate function
where $f(\mec x)$ is the number of words, $w$, in $L$ with 
${\rm occur}(w)=\mec x$, then $f(\mec x)$ is not generally polyexponential:
indeed, if $B$ has one vertex and two whole-loops, then for
fixed $s$ and large $t$,
$$
f(s,t) = 
\bigl( 1 + O_s(1/t) \bigr) \binom{s+t}{s} 2^{s-2} 2^t
= 2^t t^s \bigl( 2^{s-2}/s! + O_s(1/t) \bigr)
$$
which cannot be polyexponential, due to the $t^s$ term.
By contrast, Corollary~\ref{co_regular_unweighted} implies that
the sum of $f(s,t)$ over all $s+t=k$ is a polyexponential function of $k$.

\section{The Wording Summation Formula}
\label{se_wording_sum}

In this section we prove a lemma that is the main ingredient in the
proof of Lemma~\ref{le_second_main_lemma}.

\begin{definition}
Let $B$ be a graph and $T$ a graph.  For each $B$-wording, $W$, of $T$
we define the {\em matrix of $E_B,E_T$ occurrences in $W$}
to be the matrix
$X=X(W)\from E_B\times E_T\to\integers_{\ge 0}$ such that
$X(e_B,e_T)$ is the number of occurrences of $e_B$ in either direction
in $W(e_T')$ where $e_T'\in \Edir_T$ is either orientation of $e_T$.
\end{definition}

\begin{lemma}\label{le_most_of_the_work}
Let $B$ be a graph and $T^{\rm type}=(T,\cR)$ a $B$-type.  Let
$P=P(X)$ be a polynomial over a $E_B\times E_T$ matrix of indeterminates.
For each $\mec k\from E_T\to\naturals$, let
$$
{\rm Word}(\cR,\mec k) \eqdef \{ W \ | \ \mbox{$W$ is a $B$-wording of 
$T$ of edge-lengths $\mec k$} \}
$$
and let
$$
f(\mec k) \eqdef \sum_{W\in {\rm Word}(\cR,\mec k)} P\bigl( X(W) \bigr).
$$
Then $f$ is a polyexponential function of $\mec k$ whose bases are
some subset of the eigenvalues of $\cR$.
\end{lemma}
\begin{proof}
By linearity,
it suffices to prove this in the case where $M$ is an $E_B\times E_T$
matrix of non-negative integers, and $P(X)=X^M$ using the ``tensor notation''
$$
X^M \eqdef \prod_{e_B\in E_B,\ e_T\in E_T} X(e_B,e_T)^{M(e_B,e_T)}.
$$
So fix such an $M$.

Choosing an orientation, $\Eor_T$, of $T$, there is a natural bijection
$$
{\rm Word}(\cR,\mec k) = 
\prod_{e_T\in E_T}  \cR(e_T)_{k(e_T)},
$$
where $\cR(e_T)_{k(e_T)}$ are the words of $\cR(e_T)$ of length $k(e_T)$.
Hence
$$
\sum_{W\in {\rm Word}(\cR,\mec k)} X^M
=
\prod_{e_T\in E_T} \sum_{w\in\cR(e_T)_{k(e_T)}} 
{\rm occur}(w)^{{\rm Col}(M,e_T)},
$$
where ${\rm Col}(M,e_T)$ is the column of $M$ corresponding to $e_T$.
Corollary~\ref{co_regular_unweighted} implies that each factor in the
right-hand-side product is a polyexponential function in $k(e_T)$, 
whose bases are some subset of the eigenvalues of $\cR(e_T)$.
Hence this product is a multivariate polyexponential function of
$\mec k$, whose bases are some subset of the eigenvalues of the type
$(T,\cR)$.
\end{proof}



\section{Proofs of Lemmas~\ref{le_precursor} and \ref{le_second_main_lemma} and
Theorem~\ref{th_main_certified_walks}}
\label{se_main_cert_walks}

In this section we finish the proofs of 
Lemmas~\ref{le_precursor} and \ref{le_second_main_lemma},
from which we easily prove Theorem~\ref{th_main_certified_walks}.

\begin{proof}[Proof of Lemma~\ref{le_precursor}]
In the range $\mec k\cdot\mec 1\le n^{1/2}/C$ (and for $n\in N$) 
we have
\begin{equation}\label{eq_our_goal_second_lemma}
f(\mec k,n)
=
\sum_{W\in T^{\rm type}[\mec k]}
\Bigl( p_0(\mec a_W)+\cdots+p_{r-1}(\mec a_W)/n^{r-1} +
O(1/n^r)
g(\mec k\cdot\mec 1) \Bigr)
\end{equation} 
for a function $g$ of growth $1$.
It follows that 
$\mec a_W$ equals the vector of row sums of the $E_B\times E_T$ matrix,
$X(W)$, of occurrences of $W$.  Hence each $p_i(\mec a_W)$ is a polynomial
in $X(W)$, so exchanging summations and applying
Lemma~\ref{le_most_of_the_work} show that
\eqref{eq_our_goal_second_lemma} equals
\begin{equation}\label{eq_asymptotic_for_second_lemma}
c_0(\mec k)+\cdots+c_{r-1}(\mec k)/n^{r-1} + 
O(1/n^r)
g(\mec k\cdot\mec 1) \bigl( \#  T^{\rm type}[\mec k] \bigr)
\end{equation} 
where the $c_i(\mec k)$ are polyexponentials whose bases are the
eigenvalues of $\cR$ (i.e., of $\cR(e)$ with $e$ varying over $\Edir_T$).  
Of course, if $p_i=0$ for any $i$, then the resulting $c_i$, which is 
a summation of values of $p_i$, vanishes.

Since $g$ is of growth $1$, it suffices to 
show that
\begin{equation}\label{eq_number_of_type_T_length_k}
\#  T^{\rm type}[\mec k] 
\end{equation} 
is a function of growth $\mu_1(B)$ in $\mec k$
(since clearly the product of a function of growth $\rho_1$ and
one of growth $\rho_2$ is a function of growth $\rho_1\rho_2$).

To bound \eqref{eq_number_of_type_T_length_k},
we use the crude estimate that the total number of non-backtracking 
walks of 
length $s$ in $B$ is bounded by the sum of all entries in 
$H_B^s$, which is a function $\tilde g(s)$ of growth $\mu_1(B)$ by
Jordan canonical form.  Hence
$$
\#  T^{\rm type}[\mec k] \le
\prod_{e\in E_T} \tilde g\bigl( k(e) \bigr)
$$
which for any $\epsilon>0$ is bounded by
$$
\prod_{e\in E_T} 
\Bigl( C(\epsilon) \bigl(\mu_1(B)+\epsilon\bigr)^{k(e)} \Bigr)
\le C(\epsilon)^{\# E_T} \bigl(\mu_1(B)+\epsilon\bigr)^{\mec k\cdot\mec 1} .
$$
Hence \eqref{eq_number_of_type_T_length_k} is of growth
$\mu_1(B)$.
\end{proof}

\begin{proof}[Proof of Lemma~\ref{le_second_main_lemma}]
By definition of an algebraic model, there exist
a finite number of ordered $B$-types 
$T_j^{\rm Otype}=(T^\og,\cR_j)$ such that
(1) every $S_\Bg^\og$ of homotopy type $T^\og$ is of ordered $B$-type
$T_j^{\rm Otype}$,
(2) $\cC_n(B)$ is algebraic when restricted to $T_j^{\rm Otype}$.
Hence
$$
F_1(\mec k,n) = 
\EE_{G\in\cC_n(B)}\Bigl[\subgr_B\bigl(T^\og,\mec k;G_\Bg)\Bigr]
$$
is the sum over all $j$ of
$$
\EE_{G\in\cC_n(B)}\Bigl[\subgr_B\bigl(T_j^{\rm Otype},\mec k;G_\Bg)\Bigr].
$$
Now we apply Lemma~\ref{le_precursor} to each $j$ and take the sum.
Since any finite sum of polyexponentials over a set of bases is
again such a function, and similarly for
functions of growth $\mu_1(B)$, 
Lemma~\ref{le_second_main_lemma} follows.
\end{proof}

\begin{proof}[Proof of Theorem~\ref{th_main_certified_walks}]
As mentioned in Section~\ref{se_defs_review},
we may assume $\mu_1(B)\ge 1$, and hence
$\nu\ge 1$
(otherwise $B$ contains no SNBC walks, so neither does any
$G\in\cC_n(B)$).

According to Lemma~\ref{le_length_mult}, we have
\eqref{eq_certified_type} is the certified dot convolution of
$F_1(\mec k,n)$ with $F_2(\mec m)$ given as in
\eqref{eq_define_F_1}.
According to Lemma~\ref{le_second_main_lemma}, in the
range $\mec k\cdot\mec 1\le n^{1/2}/C$, $F_1(\mec k,n)$ is a
sum of functions $c_i(\mec k)/n^i$, plus a function
$O(1/n^r)C_r(\mec k)$,
where each $c_i(\mec k)$ and $C_r(\mec k)$ are
polyexponential functions whose bases are a set of eigenvalues of 
$T^\og$ with respect to $\cC_n(B)$.
We easily see that the certified dot convolution
$f\star_{\ge\bec\xi}g$
is bilinear in $f$ and $g$;
the linearity in $f$ implies that
\eqref{eq_certified_type} equals the sum
$$
(c_0\star_{\ge \bec\xi} g)(k)+\cdots+
(c_{r-1}\star_{\ge \bec\xi} g)(k)/n^{r-1}
+ O(1/n^r) (C_r\star_{\ge \bec\xi} g)(k)
$$
Now we apply Lemmas~\ref{le_first_main_lemma} and \ref{le_third_main_lemma}
to conclude that each function
$(c_i\star_{\ge \bec\xi} g)(k)$
and $(C_r\star_{\ge \bec\xi} g)(k)$ are $(B,\nu)$-bounded functions,
whose bases are some subset of the eigenvalues of $\cR$.
\end{proof}

\section{Pairs and Their Homotopy Type}
\label{se_pairs_prelim}

In this section we introduce some preliminary notation and ideas that
we will use to adapt the proof of
Theorem~\ref{th_main_certified_walks}
to prove Theorem~\ref{th_main_certified_pairs}.
[Theorem~\ref{th_main_certified_pairs}, 
in the case where $\psi_B^\og$ is the empty graph,
reduces to Theorem~\ref{th_main_certified_walks}.]

Our proof notes that
\begin{equation}\label{eq_ord_subgr_times_walks}
\bigl(\#[\psi_\Bg^\og]\cap G_B \bigr) 
\snbc(T^\og;\ge\bec\xi,G_B,k) 
\end{equation} 
in \eqref{eq_subgraphs_times_walks} equals 
\begin{equation}\label{eq_pairs_of_interest}
\#\bigl\{ (w,\tilde S_\Bg^\og) \;\bigm|\;  
\mbox{$w\in\SNBC(T^\og;\ge\bec\xi,G_B,k)$ and
$\tilde S_\Bg\subset G_B$ with
$\tilde S_\Bg^\og\isom \psi_\Bg^\og$} \bigr\}  ,
\end{equation} 
which is, in other words, the number of pairs $(w,\tilde S_\Bg^\og)$ where 
\begin{enumerate}
\item $w$ is an SNBC walk of length $k$ in $G_\Bg$ such
that $S_\Bg^\og=\ViSu_\Bg^\og$ lies in
$$
\subgr(T^{\rm type},\ge\bec\xi,G_\Bg,k) =
\bigcup_{\mec k\ge\bec\xi} \ \ %
\bigcup_{[S_\Bg^\og]\in T^{\rm type}[\mec k]} 
[S_\Bg^\og]\cap G_\Bg
$$
and
\item
$\tilde S_\Bg^\og$
is an ordered $B$-graph such that $\tilde S_\Bg^\og$ is isomorphic to
$\psi_\Bg^\og$ and $\tilde S_\Bg$ is a subgraph of $G_\Bg$.
\end{enumerate}
The proof of Lemma~\ref{le_length_mult}, specifically
\eqref{eq_expected_walks_fixed_vs},
immediately implies the following formula.

\begin{lemma}\label{le_length_mult_one_G}
Let $B$ be a graph, $\psi_\Bg^\og$ an ordered $B$-graph,
and $T^\og$ an ordered graph.
Then for every $B$-graph, $G_\Bg$, 
and $\bec\xi\from E_T\to\naturals$ we have 
\eqref{eq_pairs_of_interest}
(and \eqref{eq_ord_subgr_times_walks}) equals
\begin{equation}\label{eq_length_mult_pairs}
\sum_{\mec k\ge\bec\xi}
\ \sum_{\mec k\cdot \mec m=k} 
\left( 
\sum_{[S_\Bg^\og]\in T^\og[\mec k]}
{\rm pairs}\bigl(G_\Bg,S_\Bg^\og,\psi_\Bg^\og\bigr)
\right)
\legal(T^\og,\mec m)
\end{equation} 
where
\begin{align}
\nonumber
& {\rm PAIRS}\bigl(G_\Bg,S_\Bg^\og,\psi_\Bg^\og\bigr) \\ 
\label{eq_pairs_formula_def}
& \eqdef 
\biggl\{
\Bigl( (G_1)_\Bg^\og,(G_2)_\Bg^\og \Bigr) 
\ \Bigm| \ %
(G_1)_\Bg,(G_2)_\Bg\subset G_\Bg,
\ (G_1)_\Bg^\og\isom S_\Bg^\og, 
\ (G_2)_\Bg^\og\isom \psi_\Bg^\og 
\biggr\}
\end{align}
and whose cardinality is denoted by replacing ``PAIRS'' with ``pairs.''
\end{lemma}
To prove Theorem~\ref{th_main_certified_pairs} we will take
$G\in\cC_n(B)$-expected values in
\eqref{eq_length_mult_pairs} and prove that it has
an appropriate asymptotic expansion in powers of $1/n$.
To do this we will define the notion of the {\em homotopy type 
of a pair} (or simply {\em pair homotopy type}), such that
(1) every pair is of a unique homotopy type, (2) for fixed $T^\og$ and
$\psi_\Bg^\og$ there are only finitely many possible
pair homotopy types, and (3) the methods of
proving Theorem~\ref{th_main_certified_walks} generalize
easily to prove that the $\cC_n(B)$ expected value of the number
of pairs of a given homotopy type and edge-length constraints
have the desired asymptotic expansions.

The point of this section is to define a reasonable notion
of the homotopy type of a pair and to develop some of its properties.

If $\psi_\Bg^\og$ above is the empty graph, denoted $\emptyset_\Bg^\og$,
then we will see in Subsection~\ref{su_pairs_empty_graph}
that
$$
{\rm PAIRS}\bigl(G_\Bg,S_\Bg^\og,\emptyset_\Bg^\og\bigr) 
\isom
[S_\Bg^\og]\cap G_\Bg.
$$
In this sense, Theorem~\ref{th_main_certified_pairs} will reduce
to Theorem~\ref{th_main_certified_walks} in the case where
$\psi_\Bg^\og=\emptyset_\Bg^\og$; to understand the technicalities
regarding pairs, it is helpful to check
that all aspects of our proof of Theorem~\ref{th_main_certified_pairs}
reduce to that of Theorem~\ref{th_main_certified_walks}
in the case $\psi_\Bg^\og=\emptyset_\Bg^\og$.

\subsection{Motivating Remarks on the Homotopy Type of a Pair
and the Proof of Theorem~\ref{th_main_certified_pairs}}

The methods of Section~\ref{se_proof_outline},
namely Lemma~\ref{le_first_main_lemma} and \ref{le_third_main_lemma}, 
show that to
prove Theorem~\ref{th_main_certified_pairs} it suffices to
prove the analog of Lemma~\ref{le_second_main_lemma} for the function
\begin{equation}\label{eq_F_3}
F_3(\mec k,n) \eqdef
\sum_{[S_\Bg^\og]\in T^{\rm type}[\mec k]}
{\rm pairs}\bigl(G_\Bg,S_\Bg^\og,\psi_\Bg^\og\bigr)
\end{equation}
Therefore we will define the {\em homotopy type} of
a pair $( (G_1)_\Bg^\og, (G_2)_\Bg^\og )$ in a way that
we can count all such pairs as we do for the homotopy type
of $(G_1)_\Bg^\og=\ViSu_\Bg^\og(w)$ for a walk, $w$, in a $B$-graph;
of course, the problem is that
$G_1$ and $G_2$ may intersect in a complicated fashion in
$G$.
However, 
since $G_2$ is isomorphic to a fixed graph, $\psi$, it has a bounded
number of vertices and edges; so our homotopy type will
remember all of $(G_2)_\Bg^\og$ in this homotopy type, since this
is a finite amount of information; otherwise we will suppress all
vertices in $(G_1)^\og$ that are not vertices of $G_2$ and would
otherwise be suppressed in the homotopy type of $(G_1)^\og$,
i.e., all beads of $G_1$ that are not the first or last 
vertices of $G_1$.
Of course, we need to remember how $G_1,G_2$ sit inside of $G$
and how they intersect.

In this section we give one way to do the above and define a reasonable
notion of the {\em homotopy type} of a {\em pair}; there are undoubtedly
a number of possible variants.
Roughly speaking our notion of {\em pair homotopy type} has the following
properties:
\begin{enumerate}
\item
as $(G_1)_\Bg^\og$ varies over all graphs of a given $B$-type, $T^{\rm type}$,
and $(G_2)_\Bg^\og$ varies over all graphs isomorphic to a fixed
graph $\psi_\Bg^\og$,
each pair $( (G_1)_\Bg^\og, (G_2)_\Bg^\og )$ must be of a unique
``pair homotopy type,'' $X^{\rm pairH}$,
for some {\em finite} number of possible pair homotopy types;
\item
each pair homotopy type $X^{\rm pairH}$ has an ``underlying graph''
$X^\og$ such that 
each pair $( (G_1)_\Bg^\og, (G_2)_\Bg^\og )$ of this pair homotopy type
induces a $B$-wording on $X$;
\item
the set of all such wordings on $X^\og$ can be written as a
disjoint union of 
``$B$-pair types,'' $(X^\og,\cR')$, where $(X^\og,\cR')$ is an
ordered $B$-type.
\end{enumerate}
In this case all ordered $B$-graphs of homotopy type $X^\og$ can be
written as a disjoint union of 
ordered $B$-types $(X^\og,\cR)$ where the
the model is algebraic; for a fixed $\cR$ we consider all the
$B$-pair types, $(X^\og,\cR')$ of $B$-pair types above, and apply
Theorem~\ref{th_main_certified_walks} to the $B$-type $(X^\og,\cR'\cap\cR)$
where $\cR'\cap\cR$ is the $B$-type on $X$ given as 
$$
(\cR'\cap\cR)(e) = \cR'(e)\cap \cR(e).
$$

Defining the homotopy type of a pair seems to require some care to
which we now attend.
Let us begin with some preliminary definitions and remarks.

\subsection{Packaged Pairs of Ordered Graphs}

\begin{definition}
By a {\em packaged pair of ordered graphs}, we mean
a triple $\cU=(U;G_1^\og,G_2^\og)$ such that $G_1,G_2\subset U$
and $U=G_1\cup G_2$; we refer to $U$ as the {\em union} of 
$\cU$.
By a {\em morphism} from $\cU$ to another packaged pair
$\cU'=(U';{G_1'}^\og,{G_2'}^\og)$ we mean a morphism of graphs
$U\to U'$ such that for $i=1,2$, the morphism restricted to $G_i$ yields
an order preserving morphism
$G_i^\og\to {G_i'}^\og$.
For any graph $B$ we similarly define a {\em packaged
pair of ordered $B$-graphs}
$\cU_\Bg=(U_\Bg;{G_1}_\Bg,{G_2}_\Bg)$.
\end{definition}
It is easy to see that any $\cU$ as above has only
one automorphism [indeed, such a morphism restrict to the identity on
$V_{G_i}\subset V_U$ for $i=1,2$ and hence be the identity on $V_U$,
similarly for $E_U$, and, by definition, be the identity on the 
orientation].
It follows that there is at most one isomorphism $\cU\to \cU'$
for any $\cU,\cU'$ as above.

It is instructive to note that, more generally,
one can similarly define a 
{\em packaged $k$-tuple of ordered graphs}, $(U;G_1^\og,\ldots,G_k^\og)$, for
any $k\in\naturals$, and define morphisms thereof, and
similarly each such $k$-tuple has only the identity morphism as an
automorphism.  The case $k=2$ is the above case, and the case $k=1$
is a tuple $(U;G_1^\og)$ with $U=G_1$, which is effectively
just a single ordered graph, $G_1^\og$.
We also mention that a packaged pair of ordered graphs,
$(U;G_1^\og,\emptyset_\Bg^\og)$ can be identified with a $1$-tuple
of ordered graphs $(U;G_1^\og)$, and similarly for pairs of 
ordered $B$-graphs.

\subsection{Wordings, Reduction, and Homotopy Type for Pairs}

\begin{definition}\label{de_union_of_a_pair}
Let $B$ be a graph and $G_\Bg$ be a $B$-graph.
By a {\em walk-subgraph pair}
in $G_\Bg$ we mean a pair $P=(w,\tilde S_\Bg^\og)$ such that
$w$ is an SNBC walk in $G_\Bg$ and $\tilde S_\Bg^\og$ is an ordered graph such
that $\tilde S_\Bg$ is a subgraph of $G_\Bg$.
By the {\em union} of $P$ we mean 
$U_\Bg=\ViSu_\Bg(w)\cup \tilde S_\Bg$ (i.e., the smallest subgraph of
$G_\Bg$ containing both $\ViSu_\Bg(w)$ and $\tilde S_\Bg$);
we refer to $(U_\Bg;\ViSu_\Bg^\og(w),\tilde S_\Bg^\og)$
as the packaged pair of ordered $B$-graphs {\em associated to} $P$.
\end{definition}

One could allow $w$ in the above definition to be merely non-backtracking;
however, we will only be interested in SNBC walks so we limit our
discussion to them.

Let us first define the reduction and homotopy type in terms of
suppression, and then give the alternative description in terms of 
wordings on a minimal pair of ordered graphs.

\begin{definition}\label{de_pair_reduction}
Let $B$ be a graph, and $G$ be a $B$-graph.
Let $P=(w,\tilde S_\Bg^\og)$ be a walk-subgraph
pair in a $B$-graph, $G_\Bg$;
let $S_\Bg^\og=\ViSu_\Bg^\og(w)$ and
$(U;S_\Bg^\og,\tilde S_\Bg^\og)$ be its associated packaged 
pair.
By the {\em reduction} of $(w,\tilde S_\Bg^\og)$ we mean
the pair of ordered graphs $(U/V';S^\og/V',\tilde S_\Bg^\og)$,
where $V'\subset V_U$ is the subset of vertices of $V_S$ that
(1) do not lie in $V_{\tilde S}$, (2) that are of degree two in
$S$ (or, equivalently, in $U$) and not incident upon any self-loop,
and (3) are not the first vertex of $S^\og$.
We say that $P$ {\em induces the $B$-wording} on $S$ that $w$ induces on
$S/V'$ and we say that the edge-lengths of $w$ are those of $w$
on $S/V'$.
\end{definition}

Given the above definitions, 
we have a bunch of formalities that result.
The first regards Lemma~\ref{le_length_mult_one_G}, in that
we may count pairs $P=(w,\tilde S_\Bg^\og)$ in a $B$-graph by 
grouping together all $w$ with the same ordered visited $B$-subgraph
$S_\Bg^\og$ and using
\eqref{eq_length_mult_pair}.

\begin{definition}
Let $B$ be a graph, and 
$(U_\Bg;S_\Bg^\og,\tilde S_\Bg^\og)$ be a packaged pair of $B$-graph
$S_\Bg^\og=\ViSu_\Bg^\og(w)$ for some SNBC walk, $w$, in $U_\Bg$.
Then we define the {\em union}, {\em reduction}, {\em induced wording},
and {\em edge lengths} of the tuple
$(S_\Bg,\tilde S_\Bg)$ to be that of $(w,\tilde S_\Bg)$ in
Definitions~\ref{de_union_of_a_pair}
and~\ref{de_pair_reduction} (which are clearly dependent only on
$S_\Bg^\og$ rather than the particular $w$ with $S_\Bg^\og=\ViSu_\Bg^\og(w)$).
\end{definition}

An example of the reduction of packaged pair
$(U_\Bg;S_\Bg^\og,\tilde S_\Bg^\og)$ is given in 
Figure~\ref{fi_pair_reduction}, and we now make some remarks
on this figure.

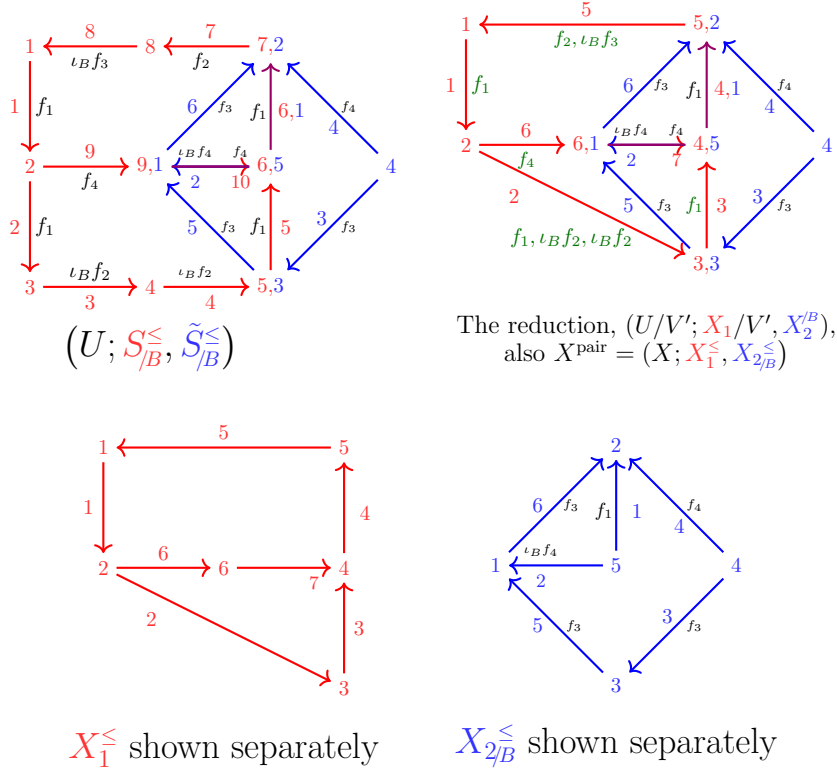
\begin{figure}
\begin{tikzpicture}[thick,scale=0.4, every node/.style={scale=0.8}]
\colorlet{bothcol}{red!60!blue}
\node at (0,-6){
  \huge 
  $\bigl(U;{\red S_\Bg^\og},{\blue \tilde S_\Bg^\og}\bigr)$
  } ;
\node at (-4,4)(1){\red $1$};
\node at (-4,0)(2){\red $2$};
\node at (-4,-4)(3){\red $3$};
\node at (0,-4)(4){\red $4$};
\node at (4,-4)(5){\red $5$,\blue $3$};
\node at (4,0)(6){\red $6$,\blue $5$};
\node at (4,4)(7){\red $7$,\blue $2$};
\node at (0,4)(8){\red $8$};
\node at (0,0)(9){\red $9$,\blue $1$};
\node at (8,0)(10){\blue $4$};
\draw[->,red] (1) -- (2) ;
\draw[->,red] (2) -- (3) ;
\draw[->,red] (3) -- (4) ;
\draw[->,red] (4) -- (5) ;
\draw[->,red] (5) -- (6) ;
\draw[->,bothcol] (6) -- (7) ;
\draw[->,red] (7) -- (8) ;
\draw[->,red] (9) -- (6) ;
\draw[->,blue] (6) -- (9) ;
\draw[bothcol] (6) -- (9) ;
\draw[->,red] (2) -- (9) ;
\draw[->,red] (8) -- (1) ;
\draw[->,blue] (10) -- (7) ;
\draw[->,blue] (10) -- (5) ;
\draw[->,blue] (9) -- (7) ;
\draw[->,blue] (5) -- (9) ;
\node at (-4.5,2){\red $1$} ;
\node at (-3.5,2){$f_1$} ;
\node at (-4.5,-2){\red $2$} ;
\node at (-3.5,-2){$f_1$} ;
\node at (-2,-4.5){\red $3$} ;
\node at (-2,-3.5){$\iota_B f_2$} ;
\node at (2,-4.5){\red $4$} ;
\node at (1.5,-3.5){\tiny$\iota_B f_2$} ;
\node at (4.5,-2){\red $5$} ;
\node at (3.6,-2){\small $f_1$} ;
\node at (4.7,1.8){\red $6$,\blue $1$} ;
\node at (3.6,1.8){\small $f_1$} ;
\node at (2,4.5){\red $7$} ;
\node at (1.7,3.5){\small $f_2$} ;
\node at (-2,4.5){\red $8$} ;
\node at (-2,3.5){\small $\iota_B f_3$} ;
\node at (-2,0.5){\red $9$} ;
\node at (-2,-0.5){\small $f_4$} ;
\node at (3,0.5){\tiny $f_4$} ;
\node at (3,-0.5){\small\red $10$} ;
\node at (1.5,0.5){\tiny $\iota_B f_4$} ;
\node at (1.5,-0.5){\small\blue $2$} ;
\node at (2.5,2){\tiny $f_3$} ;
\node at (1.4,2){\blue $6$} ;
\node at (2.6,-2){\tiny $f_3$} ;
\node at (1.4,-2){\blue $5$} ;
\node at (6.6,2){\tiny $f_4$} ;
\node at (6.1,1.3){\blue $4$} ;
\node at (6.6,-2){\tiny $f_3$} ;
\node at (5.7,-1.7){\blue $3$} ;
\end{tikzpicture}
\quad%
\begin{tikzpicture}[thick,scale=0.4, every node/.style={scale=0.8}]
\colorlet{bothcol}{red!60!blue}
\node at (2,-6){
  \Large The reduction, $(U/V';{\red X_1}/V',{\blue X_2^\Bg})$, };
\node at (2,-7){
  \Large also $X^{\rm pair}=(X;{\red X_1^\og},{\blue {X_2}_\Bg^\og}\bigr)$} ;
\node at (-4,4)(1){\red $1$};
\node at (-4,0)(2){\red $2$};
\node at (4,-4)(5){\red $3$,\blue $3$};
\node at (4,0)(6){\red $4$,\blue $5$};
\node at (4,4)(7){\red $5$,\blue $2$};
\node at (0,0)(9){\red $6$,\blue $1$};
\node at (8,0)(10){\blue $4$};
\draw[->,red] (1) -- (2) ;
\draw[->,red] (2) -- (5) ;
\draw[->,red] (5) -- (6) ;
\draw[->,bothcol] (6) -- (7) ;
\draw[->,red] (7) -- (1) ;
\draw[->,red] (9) -- (6) ;
\draw[->,blue] (6) -- (9) ;
\draw[bothcol] (6) -- (9) ;
\draw[->,red] (2) -- (9) ;
\draw[->,blue] (10) -- (7) ;
\draw[->,blue] (10) -- (5) ;
\draw[->,blue] (9) -- (7) ;
\draw[->,blue] (5) -- (9) ;
\node at (-4.5,2){\red $1$} ;
\node at (-3.5,2){\darkgreen $f_1$} ;
\node at (-2.4,-1.7){\red $2$} ;
\node at (-0.5,-3.1){\darkgreen $f_1,\iota_B f_2,\iota_B f_2$} ;
\node at (4.5,-2){\red $3$} ;
\node at (3.6,-2){\darkgreen \small $f_1$} ;
\node at (4.7,1.8){\red $4$,\blue $1$} ;
\node at (3.6,1.8){\small $f_1$} ;
\node at (0,4.5){\red $5$} ;
\node at (0,3.5){\darkgreen \small $f_2,\iota_B f_3$} ;
\node at (-2,0.5){\red $6$} ;
\node at (-2,-0.5){\darkgreen \small $f_4$} ;
\node at (3,0.5){\tiny $f_4$} ;
\node at (3,-0.5){\small\red $7$} ;
\node at (1.5,0.5){\tiny $\iota_B f_4$} ;
\node at (1.5,-0.5){\small\blue $2$} ;
\node at (2.5,2){\tiny $f_3$} ;
\node at (1.4,2){\blue $6$} ;
\node at (2.6,-2){\tiny $f_3$} ;
\node at (1.4,-2){\blue $5$} ;
\node at (6.6,2){\tiny $f_4$} ;
\node at (6.1,1.3){\blue $4$} ;
\node at (6.6,-2){\tiny $f_3$} ;
\node at (5.7,-1.7){\blue $3$} ;
\end{tikzpicture}
\vskip 0.2truein
\begin{tikzpicture}[thick,scale=0.4, every node/.style={scale=0.8}]
\colorlet{bothcol}{red!60!blue}
\node at (0,-6){
  \huge ${\red X_1^\og}$ shown separately} ;
\node at (-4,4)(1){\red $1$};
\node at (-4,0)(2){\red $2$};
\node at (4,-4)(5){\red $3$};
\node at (4,0)(6){\red $4$};
\node at (4,4)(7){\red $5$};
\node at (0,0)(9){\red $6$};
\draw[->,red] (1) -- (2) ;
\draw[->,red] (2) -- (5) ;
\draw[->,red] (5) -- (6) ;
\draw[->,red] (6) -- (7) ;
\draw[->,red] (7) -- (1) ;
\draw[->,red] (9) -- (6) ;
\draw[->,red] (2) -- (9) ;
\node at (-4.5,2){\red $1$} ;
\node at (-2.4,-1.7){\red $2$} ;
\node at (4.5,-2){\red $3$} ;
\node at (4.7,1.8){\red $4$} ;
\node at (0,4.5){\red $5$} ;
\node at (-2,0.5){\red $6$} ;
\node at (3,-0.5){\small\red $7$} ;
\end{tikzpicture}
\quad\quad%
\begin{tikzpicture}[thick,scale=0.4, every node/.style={scale=0.8}]
\node at (4,-6){
  \huge $\blue{{X_2}_\Bg^\og}$ shown separately} ;
\node at (4,-4)(5){\blue $3$};
\node at (4,0)(6){\blue $5$};
\node at (4,4)(7){\blue $2$};
\node at (0,0)(9){\blue $1$};
\node at (8,0)(10){\blue $4$};
\draw[->,blue] (6) -- (7) ;
\draw[->,blue] (6) -- (9) ;
\draw[->,blue] (10) -- (7) ;
\draw[->,blue] (10) -- (5) ;
\draw[->,blue] (9) -- (7) ;
\draw[->,blue] (5) -- (9) ;
\node at (4.7,1.8){\blue $1$} ;
\node at (3.6,1.8){\small $f_1$} ;
\node at (1.5,0.5){\tiny $\iota_B f_4$} ;
\node at (1.5,-0.5){\small\blue $2$} ;
\node at (2.5,2){\tiny $f_3$} ;
\node at (1.4,2){\blue $6$} ;
\node at (2.6,-2){\tiny $f_3$} ;
\node at (1.4,-2){\blue $5$} ;
\node at (6.6,2){\tiny $f_4$} ;
\node at (6.1,1.3){\blue $4$} ;
\node at (6.6,-2){\tiny $f_3$} ;
\node at (5.7,-1.7){\blue $3$} ;
\end{tikzpicture}
\caption{A packaged pair and its homotopy type}
\label{fi_pair_reduction}
\end{figure}

Figure~\ref{fi_pair_reduction} depicts a $B$-graph
$S_\Bg^\og$ (in red) with $9$ vertices and $10$ edges, whose ordering
is depicted (in red) by numbering the vertices, the edges, and showing
the orientation of each edge; $S_\Bg^\og$ is necessarily ordered
as the first encountered ordering of an SNBC walk.
We do not describe $B$ in the figure, but $B$ has four directed
edges $f_1,\ldots,f_4$ 
which are whole-loops about a single vertex;
$B$ could have additional vertices and edges;
the structure map $S\to B$ is indicated by writing (in black)
an $f_i$ or $\iota_B f_i$.
By contrast, $\tilde S_\Bg^\og$ (depicted in blue)
has an arbitrary ordering.
$S$ and $\tilde S$ share two edges (depicted in purple): 
the horizontal edge they
share is ordered differently: in $S_\Bg^\og$, the $10$-th edge
is traversed from its $9$-th vertex to its $6$-th vertex,
whereas this same edge is oriented in the other direction in
$\tilde S_\Bg^\og$;
the vertical edge they share is oriented in the same direction.
In the picture of $(U;S_\Bg^\og,\tilde S_\Bg^\og)$ we have
kept the wording in dark green (see below) that $S_\Bg^\og$ induces on
the directed edges that lie only in $S_\Bg^\og$ and not
$\tilde S_\Bg^\og$, even though $X_1^\og$ does not include this
data in the pair type $X^{\rm pair}=(X;X_1^\og,{X_2}_\Bg^\og)$.
Note that $X_1^\og$ has three vertices of degree two, namely its first
vertex and two vertices that it shares with $X_2$.

Of course, we treat $S_\Bg^\og$ very differently from
$\tilde S_\Bg^\og$ in defining the pair homotopy type of
a packaged pair $(U;S_\Bg^\og,\tilde S_\Bg^\og)$ because of our
application: namely we apply this notion to the situation 
where $\tilde S_\Bg^\og$ is in a fixed isomorphism class of
ordered $B$-graphs, while $S_\Bg^\og=\ViSu(w)$ for an SNBC walk
varying over an entire homotopy class (or perhaps an entire ordered
$B$-type).

Now we say what we mean by a {\em pair homotopy type} and for
walk-subgraph pairs or packaged pairs of $B$-graphs to be
{\em of} such a pair homotopy type.

\begin{definition}
Let $B$ be a graph.
By a {\em pair homotopy type (over $B$)} we mean any triple 
$X^{\rm pairH}=(X;X_1^\og,{X_2}^\og_\Bg)$ where
$(X;X_1^\og,X_2^\og)$ is a packaged pair of ordered graphs and 
$X_2$ is endowed
with the structure of a $B$-graph.
We say that $X^{\rm pairH}$ is {\em isomorphic} to another pair homotopy type
$Y^{\rm pairH}=(Y;Y_1^\og,{Y_2}^\og_\Bg)$ if
there exists an isomorphism of packaged pairs
from $(X;X_1^\og,X_2^\og)$
to $(Y;Y_1^\og,Y_2^\og)$ such that this unique isomorphism,
given by the isomorphism $X\to Y$, restricts to an isomorphism
$X_2\to Y_2$ that respects their $B$-graph structure
(i.e., is an isomorphism $(X_2)_\Bg\to (Y_2)_\Bg$).
We say that a walk-subgraph pair in a $B$-graph 
is {\em of homotopy type $X^{\rm pairH}$} if its reduction is
isomorphic to $X^{\rm pairH}$, and similarly for 
packaged pair $(U_\Bg;S_\Bg^\og,\tilde S_\Bg^\og)$.
\end{definition}

Next we need an abstract notion of a $B$-wording and of a $B$-type
for pair homotopy types (over $B$).
The point is that for $X^{\rm pairH}=(X;X_1^\og,{X_2}^\og_\Bg)$,
the $B$-structure on $X_2$ is determined; hence these notions
reduce to the analogous notions in $X_1$, provided that they
are compatible on the $\Edir_{X_1}\cap \Edir_{X_2}$.

\begin{definition}
Let $B$ be a graph, and
$X^{\rm pairH}=(X;X_1^\og,{X_2}^\og_\Bg)$ a pair homotopy type over $B$.
By a {\em $B$-wording} on $X^{\rm pair}$ we mean a $B$-wording
$W$ on $X_1$ such that if $e\in\Edir_{X_1}\cap \Edir_{X_2}$ then
$W(e)$ is the one-letter word $e_B\in\Edir_B$ lying under $e$ in the
structure map ${X_2}_\Bg$.
The {\em edge-lengths} of a $B$-wording are its edge-lengths on $X_1$.
By a {\em $B$-type} on $X^{\rm pair}$ we mean a $B$-type,
$(X_1,\cR)$, on $X_1$, such that
if $e\in\Edir_{X_1}\cap \Edir_{X_2}$ then
$\cR(e)$ consists of one word, namely $e_B$ as above.
\end{definition}

Figure~\ref{fi_pair_reduction} illustrates the wording that
$S_\Bg^\og$ of the packaged pair
$(U_\Bg;S_\Bg^\og,\tilde S_\Bg^\og)$ induces $X_1$.
[The wording on the edges in $\tilde S_\Bg^\og$ are just the
one-letter words given by the $B$-structure $\tilde S\to B$, which
we remember as part of the data of $X^{\rm pair}$.]

It will be convenient to speak of edge-lengths and wordings
of a pair $(X;X_1^\og,{X_2}^\og_\Bg)$ as living on $X$ as opposed to
$X_1$, since ${X_2}_\Bg$ determines the edge-lengths and wordings of
each directed edge in $X$ that lies in $X_2$.
Let us formalize.

\begin{definition}
Let $B$ be a graph, and
$X^{\rm pairH}=(X;X_1^\og,{X_2}^\og_\Bg)$ a pair homotopy type over $B$.
We will identify a set of edge-lengths $\mec K\from E_{X_1}\to\naturals$
as a function $\mec K\from E_X\to\naturals$ by setting extending it via
$K(e)=1$ for $e\in E_{X_2}$.
If $W$ is a $B$-wording (respectively, $(X_1,\cR)$ a
$B$-type) on $X^{\rm pair}$, we will extend it as a $B$-wording
(respectively, $B$-type) on $X$ by setting
by setting for each $e\in\Edir_{X_2}$
the value $W(e)=e_B$ where $e_B$ is the edge under $e$ in the
$B$-structure map $X_2\to B$ (respectively, the language consisting
of the single, one-letter word $e_B$).
\end{definition}

\subsection{Pairs Involving the Empty Graph}
\label{su_pairs_empty_graph}

If a pair homotopy type has
$X^{\rm pairH}=(X;X_1^\og,{X_2}^\og_\Bg)$ has
$X_2$ equal to the {\em empty graph}, then all the definitions
in this section reduce to those for homotopy types of walks.
We carefully explain this, since we
believe one can get (mostly) good intuition for
the definitions in this section by considering this special
case.

(The reader who dislikes empty sets and empty graphs can, of
course, skip this subsection.)

There is a unique {\em empty graph}, whose vertex and directed edge sets
are the empty set, $\emptyset$,
and $h,t,\iota$ are the unique maps $\emptyset\to\emptyset$;
this graph has a unique $B$-structure and ordering, and we use
$\emptyset_\Bg^\og$ to denote this (unique) empty ordered $B$-graph.

To each SNBC walk, $w$, in a $B$-graph, we associate
the pair $(w,\emptyset_\Bg^\og)$.  All definitions in this section
regarding pairs $(w,\tilde S_\Bg^\og)$ with 
$\tilde S_\Bg^\og=\emptyset_\Bg^\og$ reduce to the corresponding
definition on $w$.

For example, the 
homotopy type of $(w,\emptyset_\Bg^\og)$ is that of a pair 
$X^{\rm pairH}=(X;X_1^\og,{X_2}^\og_\Bg)$ where 
${X_2}^\og_\Bg=\emptyset_\Bg^\og$, and the knowledge of
a pair homotopy type $(X;X_1^\og,\emptyset_\Bg^\og)$ is equivalent
to knowing $X_1^\og$ and $X$ where $X$ is the union of $X_1$ and the
empty graph.
Therefore the pair homotopy type of $(w,\emptyset_\Bg^\og)$ is equivalent
to knowing the 1-tuple of ordered graphs $(X;X_1^\og)$ which is equivalent
to knowing $X_1^\og$, which
is just the homotopy type of $w$.

Note also that Theorem~\ref{th_main_certified_walks} is 
special case of
Theorem~\ref{th_main_certified_pairs}, thanks to the empty graph:
indeed, for any $B$-graph, $G_\Bg$, we easily unwind the definitions to
see that
$[\emptyset_\Bg^\og]\cap G_\Bg$ consists of a single element.
Hence for $\psi_\Bg^\og=\emptyset_\Bg^\og$,
Theorem~\ref{th_main_certified_pairs} reduces to
Theorem~\ref{th_main_certified_walks}.

\subsection{The Homotopy Type of the Walk of a Pair}

In view of \eqref{eq_pairs_of_interest},
Theorem~\ref{th_main_certified_walks} concerns the 
expected number of walk-subgraph
pairs $(w,\tilde S_\Bg^\og)$ in a $G\in\cC_n(B)$ such that
\begin{equation}\label{eq_walk_subgraph_conditions}
w\in\SNBC(T^\og;\ge\bec\xi,G_B,k), \quad
\tilde S_\Bg^\og\isom \psi_\Bg^\og.
\end{equation} 
We now express the conditions in \eqref{eq_walk_subgraph_conditions}
as conditions in terms of the pair homotopy type
$X^{\rm pairH}=(X;X_1^\og,\psi^\og_\Bg)$ of $(w,\tilde S_\Bg^\og)$.

The main point is evident once we carefully keep track of things:
the homotopy type of $X_1$ and $w$ must be
the same, say $T^\og$, and the edge-lengths and wordings that
$w$ induces on $T^\og$ can be read off from those that
$(w,\tilde S_\Bg^\og)$ induces on $X^{\rm pair}$ (but not vice versa).
An example is given in Figure~\ref{fi_pair_reduction}:
there the vertices of $X_1^\og$ (in red) numbered $3,5,6$ are beads, which
are not suppressed in $S_1/V'$,
since they lie also in $X_2$; once we suppress these
vertices we keep only the vertices $1,2,4$, and we see that the homotopy
type of $S^\og$ (or of an SNBC walk, $w$, such that
$S^\og=\ViSu^\og(w)$) is $T^\og$ where $T$ has the homotopy type
of a theta graph and the first vertex of $T^\og$ lies in the middle
of one of the edges that forms the theta.

\begin{lemma}\label{le_homo_type_walk_pair}
Let $B$ be a graph, $T^\og$ an ordered graph, and 
$\psi_\Bg^\og$ an ordered $B$-graph.
Let $(w,\tilde S_\Bg^\og)$ be a walk-subgraph pair in
a $B$-graph, $G_\Bg$, and let $S_\Bg^\og=\ViSu_\Bg^\og(w)$.
Then the following are equivalent:
\begin{enumerate}
\item we have
$$
w\in\SNBC(T^\og;\ge\bec\xi,G_B,k), \quad
\tilde S_\Bg^\og\isom \psi_\Bg^\og,
$$
and
\item
$(w,\tilde S_\Bg^\og)$ is of
pair homotopy type
$X^{\rm pairH}=(X;X_1^\og,\psi^\og_\Bg)$ for some ordered graph
$X_1^\og$ of homotopy type $T^\og$
and some $X$ containing $X_1,\psi$ as subgraphs.
\end{enumerate}
Moreover, if these conditions hold and we fix such an $X^{\rm pair}$,
then:
\begin{enumerate}
\item
if $\mec k\from E_T\to\naturals$ are the edge-lengths that
$S_\Bg^\og$ induces on
$T^\og$, and $\mec K\from E_{X_1}\to\naturals$ are the edge-lengths that
$(w,\tilde S_\Bg^\og)$ induces on $X^{\rm pair}$,
then for each $e_T\in E_T$ we have
\begin{equation}\label{eq_edge_lengths_relation}
k(e_T) = \sum_{e_X\in e_T} K(e_X),
\end{equation} 
where $e_X\in e_T$ means that the path in $T^\og$ corresponding
to $e_T$ contains $e_X$ as one of its edges; similarly
\item
if $W_T$ is the wording $\Edir_T\to{\rm NBWALKS}(B)$ that $w$
induces on $T^\og$, and
$W_X$ the wording $\Edir_X\to{\rm NBWALKS}(B)$ that $(w,\tilde S_\Bg^\og)$
induces on $X^{\rm pair}$, then for each $e_T\in\Edir_T$ we have
\begin{equation}\label{eq_wording_relation}
w_T(e_T) = w_X(e_{X,1}) \circ \ldots \circ w_X(e_{X,s})
\end{equation} 
where $\circ$ denotes concatenation of strings, and 
$e_{X,1},\ldots,e_{X,s}$ is the beaded path in $X_1$ corresponding
to $e_T$; and
\item
each directed edge in $\Edir_X$ lies in exactly
one non-backtracking walk $e_{X,1},\ldots,e_{X,s}$ forming a directed edge
in $E_T$.
\end{enumerate}
\end{lemma}
The proof is straightforward, but a bit long
to write out carefully.
\begin{proof}
Let us prove the following more general principle:
let $V'$ be any proper set of beads of an ordered
$B$-graph $S_\Bg^\og$, and $V'',V'''$ a partition of $V'$ (i.e., $V'',V'''$
are disjoint subsets whose union is $V'$).  Let us show that
that $V''$ is a proper set of beads of $S$, $V'''$ is a proper set
of beads of $S/V''$, and there
is a natural isomorphism
\begin{equation}\label{eq_iterated_suppression}
S^\og / V' \isom (S^\og/V'')/V''' ;
\end{equation} 
moreover, setting $T=S^\og/V/$ and $X=S^\og/V''$, this isomorphism
takes a directed edge $e_T\in \Edir_T$ 
to a $V'''$-beaded path in $X$
whose directed edges are $e_{X,1},\ldots,e_{X,s}$ such that
if $W_T,W_X$ respectively denote the wordings that $S_\Bg^\og$
induces on $T,X$, then \eqref{eq_wording_relation} holds.
Then \eqref{eq_edge_lengths_relation} is an immediate consequence.

Once we verify the above principle, then Lemma~\ref{le_homo_type_walk_pair}
follows by taking $V'$ to be the set of all beads of $S_1^\og$ that
exclude the first vertex (if it is a bead), and setting
$V''=V'\cap V_{S_2}$ and $V'''=V'\setminus V''=V'\setminus V_{S_2}$.

The isomorphism \eqref{eq_iterated_suppression} and all the properties
claimed in that paragraph follow by unwinding the definitions,
which we now do.

First, since $V''\subset V'$, $V''$ is a proper bead set of $V$;
if $v\in V'''$ then $v$ is a bead in $S$, so it is incident upon
two edges of $S$, and each of those edges lie on a distinct edge
of $X=S/V''$; since passing from $S$ to $X=S/V''$ does not create any
new half-loops, $v$ is not incident upon a half-loop of $X=S/V''$.
Finally, $v$ cannot be incident upon a whole-loop $X=S/V''$, since
otherwise $V''\cup\{v\}$ would contain the entire connected component
of $v$ in $S$, contradicting the fact that $V'$ is a proper bead
set in $V$.
Hence $V'''$ is a set of beads of $X=S/V''$; similarly, $V'''$
is a proper bead subset of $X=S/V''$, or otherwise $V''\cup V'''=V'$
would contain an entire connected component of $S$,
contradicting the fact that $V'$ is a proper bead
set in $V$.

Second, we see that, by definition,
$$
V_T = V_S\setminus V' = (V_S\setminus V'')\setminus V'''
= V_X\setminus V''' = V_{X/V'''}.
$$
Hence $S/V'$ and $X/V'''=(S/V'')/V'''$ have identical vertex sets;
we define the morphism \eqref{eq_iterated_suppression} to be 
the identity from $V_T$ to $V_{X/V'''}$.

Third, a directed edge, $e_T\in \Edir_T$, is, by definition,
a $V'$-beaded path of $S$, meaning a non-backtracking walk,
$w_S$, in $S$
whose endpoints lie in $V_S\setminus V'$ and whose intermediate
vertices lie in $V'$.  Hence $w$ is a concatenation
$e_{T,1},\ldots,e_{T,s}$ of non-backtracking walks such that
$e_{T,1},\ldots,e_{T,s-1}$ terminate in vertices of $V'''$
and such that the intermediate vertices of each of $e_{T,j}$
lie in $V''$.  This gives a morphism
$$
\Edir_T \to \Edir_{X/V'''} ;
$$
conversely, we see that any element of $\Edir_{X/V'''}$
is such a non-backtracking walk $e_{T,1},\ldots,e_{T,s}$,
which gives the inverse map.

Fourth, since any directed edge of a graph lies in exactly
one directed edge of any of its suspensions, we see that
any $e_X\Edir_X$ lies in exactly one directed edge of
$E_{X/V'''}$, and hence one directed edge in $\Edir_T$.
\end{proof}

\subsection{Finiteness of Homotopy Types of Pairs}

\begin{lemma}
Let $B$ be a graph,
$T^\og$ an ordered graph, and $\psi_B^\og$ an ordered $B$-graph.
There are only finitely many possible homotopy types
$X^{\rm pairH}=(X;X_1^\og,{X_2}^\og_\Bg)$ of walk-subgraph
pairs $(w,\tilde S_\Bg^\og)$ (or, equivalently
packaged pairs $(S_\Bg^\og,\tilde S_\Bg^\og)$ with
$S_\Bg^\og=\ViSu_\Bg^\og(w)$) such that
$w$ (or $S^\og$) is of homotopy type $T^\og$ and $\tilde S_\Bg^\og$
is isomorphic to $\psi_B^\og$.
\end{lemma}
\begin{proof}
We need to produce a finite set of pair homotopy types 
$X^{\rm pairH}=(X;X_1^\og,{X_2}^\og_\Bg)$
that includes all walk-subgraph pairs (or packaged pairs) as above.
First, clearly we can take ${X_2}^\og_\Bg=\psi^\og_\Bg$.
Second, $X_1$ is obtained from $T^\og$ by introducing at most
$\#V_\psi$ vertices as beads along its edges.
Hence there are a finite number of possible $X_1^\og$, up to
isomorphism.  Then the union of $X_1^\og$ and $\psi$ has a
bounded number of vertices (i.e., bounded by $\#V_{X_1}+\#V_\psi$)
and a bounded number of directed edges, yielding a finite number
of possible graphs $X$.  Then $X_1$ and $\psi$ can be subgraphs
of $X$ in only finitely many ways, which therefore yields a finite
number of pair homotopy
types $X^{\rm pairH}=(X;X_1^\og,{X_2}^\og_\Bg)$.
\end{proof}

\section{Proof of Theorem~\ref{th_main_certified_pairs}}
\label{se_main_cert_pairs}

We now adapt the proof of Theorem~\ref{th_main_certified_walks},
using the new notions introduced in Section~\ref{se_pairs_prelim}, to
prove Theorem~\ref{th_main_certified_pairs}.

\subsection{Pairs with Certified Edge-Lengths}

In this section we give a natural generalization of
Theorem~\ref{th_main_certified_walks} to pairs with certified edge-lengths.
It will easily---although not immediately---imply
Theorem~\ref{th_main_certified_pairs}.

\begin{definition}
Let $B$ be a graph and
$X^{\rm pairH}=(X;X_1^\og,{X_2}^\og_\Bg)$ a homotopy type of a
walk-subgraph pair in a $B$-graph.
For any $\mec K\from E_X\to \naturals$, let
$$
\mbox{WS-PAIRS}(X^{\rm pairH},\mec K; G_\Bg,k)
$$
be the set of walk-subgraph pairs
in $G_\Bg$ of pair homotopy type $X^{\rm pairH}$ whose edge lengths are
$\mec K$ and whose walk has length $k$; similarly we let
$$
\mbox{PACK-PAIRS}(X^{\rm pairH},\mec K; G_\Bg)
$$
the set of packaged pairs of
$B$-subgraphs $(U_\Bg,S_\Bg^\og,\tilde S_\Bg^\og)$
in $G_\Bg$ (i.e., with $U_\Bg\subset G_\Bg$)
of homotopy type $X^{\rm pairH}$ whose edge lengths are
$\mec K$.
For
$\bec\Xi\from E_X\to\naturals$, the {\em walk-subgraph pairs 
certified by $\Xi$ (over $X^{\rm pairH}$)}
is the set 
\begin{equation}\label{eq_de_edge_cert_pairs}
\mbox{WS-PAIRS}(X^{\rm pairH},\ge\bec\Xi; G_\Bg,k) \eqdef
\bigcup_{\mec K\ge\bec\Xi}\mbox{WS-PAIRS}(X^{\rm pairH},\mec K; G_\Bg,k)  ;
\end{equation}
we similarly define the set of {\em packaged pairs certified
by by $\Xi$ (over $X^{\rm pairH}$)} to be
$$
\mbox{PACK-PAIRS}(X^{\rm pairH},\ge\bec\Xi; G_\Bg,k) \eqdef
\bigcup_{\mec K\ge\bec\Xi}\mbox{PACK-PAIRS}(X^{\rm pairH},\mec K; G_\Bg,k)  ;
$$
we use ws-pairs and pack-pairs 
instead of WS-PAIRS and pack-pairs for the cardinality of these sets.
\end{definition}

\begin{lemma}\label{le_pair_certified_thm}
Let $B$ be a graph, and $\cC_n(B)$ an algebraic model.
For any pair homotopy type
$X^{\rm pairH}=(X;X_1^\og,{X_2}^\og_\Bg)$ and any
$\bec\Xi\from E_X\to\naturals$, let
$$
\nu = \max\Bigl( \mu_1^{1/2}(B), 
\mu_1\bigl(\VLG({X_1},\bec\Xi|_{X_1})\bigr) \Bigr).
$$
Then for any $r\ge 1$,
\begin{equation}\label{eq_expected_ws_pairs}
f(k,n)\eqdef \EE_{G\in\cC_n(B)}
[\mbox{\rm ws-pairs}(X^{\rm pairH},\ge\bec\Xi;G,k)]
\end{equation} 
has a $(B,\nu)$-bounded expansion 
$$
c_0(k)+\cdots+c_{r-1}(k)+ O(1) c_r(k)/n^r,
$$
to order $r$; the bases (larger with respect to $\nu$) of
of the coefficients of the expansion is a subset of the
eigenvalues of the model, and $c_i(k)=0$ for any $i$ less than the
order of any $B$-graph that contains a walk of type $T^{\rm type}$
and a subgraph isomorphic to $\psi_\Bg$.
\end{lemma}
[In the above lemma we use $\bec\Xi|_{X_1}$ to emphasize that
we are viewing $\bec\Xi$ as restricted to
$E_{X_1}$, since the lemma views $\bec\Xi$ as defined on all of $E_X$.]

\begin{proof}
We have $\mbox{pairs}(X^{\rm pairH},\mec K;G,k)=0$ if $K(e)>1$ for 
any $e\in E_\psi$; hence for all such $e$ we may assume $\Xi(e)=1$ 
and that $\mec K\ge \bec\Xi$ in \eqref{eq_de_edge_cert_pairs} refers to
to those $\mec K$ with $K(e)=1$ for all $e\in E_\psi$.

Next we prove the following analog of Lemma~\ref{le_length_mult}:
for $f(k,n)$
as in \eqref{eq_expected_ws_pairs},
\begin{equation}\label{eq_length_mult_pair}
f(k,n) = \sum_{\mec K|_{X_1}\cdot\mec m=k, \ \mec K|_{X_1}\ge \bec\Xi|_{X_1}} 
F_1(\mec K|_{X_1},n)F_2(\mec m),
\end{equation} 
where ($\mec m\from E_{X_1}\to\naturals$ and)
\begin{align*}
F_1(\mec K|_{X_1}) =
\tilde F_1(\mec K,n) & = 
\EE_{G\in\cC_n(B)}\Bigl[
\mbox{pack-pairs}(X^{\rm pairH},\mec K,G_\Bg)\Bigr]
\\
F_2(\mec m)  & = \legal({X_1}^\og,\mec m) ,
\end{align*}
where $\legal$ is as in Definition~\ref{de_legal}, and where we
write $F_1(\mec K|_{X_1},n)=\tilde F_1(\mec K,n)$ since $K$ is
determined by its values on $X_1$.
To prove \eqref{eq_length_mult_pair} we note that
to each walk-subgraph pair in $G_\Bg$ there
corresponds a 
unique packaged pair $(U_\Bg,S_\Bg^\og,\tilde S_\Bg^\og)$,
and for each such packaged pair of edge-lengths $\mec K$,
the number of legal walks in $S^\og$ of length $k$ is given
as in Lemma~\ref{le_fundamental_legal_formula},
and $\mec k_S$ in 
\eqref{eq_fundamental_num_walks} equals $\mec K|_{X_1}$.

Now we follow the proof of Theorem~\ref{th_main_certified_walks},
where the role of $T^\og$ is played by $X^\og$.

Since $\cC_n(B)$ is algebraic, there are ordered $B$-types
$(X^\og,\cR_j)$ such that each wording $W$ of $X^\og$ lies
in exactly one of $\cR_j$; let
$\pi\from X_2\to B$ be the $B$-structure of ${X_2}_\Bg$, and
define $\cR_j'$ by
$$
\cR_j'(e) = \left\{ \begin{array}{ll}
\cR_j(e) & \mbox{if $e\in \Edir_X\setminus \Edir_{X_2}$, and} \\
\{\pi(e)\}\cap \cR_j(e) & \mbox{if $e\in \Edir_{X_2}$.} 
\end{array}
\right.
$$
Since the one-word language $\{\pi(e)\}$ has eigenvalue $0$
(its number of words of a given length vanishes for length greater
than one),
all the eigenvalues of $\cR_j'$ are either $0$ or those of $\cR$.

Applying Lemma~\ref{le_precursor} to each $(X^\og,\cR_j')$ and summing 
over $j$ shows that $\tilde F_1'(\mec K,n)$ has 
expansions to any order $r$ whose
coefficients are polyexponential functions, $p_i$, 
in $\mec a(\mec K)$, whose bases are those of $\cR_j'$,
and whose error term is $n^{-r}$ times a function
of growth $\mu_1(B)$.
Note that we may write
$$
\mec a(\mec K) = \mec a(\mec K|_{X_1}) + \mec a(\mec K|_{\psi\setminus {X_1}})
$$
where the rightmost $\mec a$ counts $E_B$ edge occurrences in
$\psi_\Bg^\og$ edges that are not in ${X_1}$; hence 
(since $X^{\rm pairH}$ is fixed), 
$\mec a(\mec K)$ is a linear function of $\mec a(\mec K|_{X_1})$,
and therefore each $p_i(\mec a(\mec K))$ also a polyexponential function
of $\mec K|_{X_1}$.  
Hence $F_1(\mec K|_{X_1},n)$ has expansions to any order
$r$, whose coefficients are polyexponentials whose bases are $0$ plus those
of the model, and whose error term is $n^{-r}$ times a function of
growth $\mu_1(B)$.
Furthermore,
Lemma~\ref{le_precursor}
implies that these coefficients $c_i$ vanish for all $i<\ord(X)$.

Lemma~\ref{le_first_main_lemma} shows that 
$$
\omega(M) \eqdef
\sum_{\bec\cert\cdot \mec m= M} F_2(\mec m)
$$
is of growth $\nu$.  Now we use Lemma~\ref{le_third_main_lemma} to 
conclude that 
\eqref{eq_length_mult_pair} has $(B,\nu)$-bounded expansions to any 
order, whose coefficients have bases that are those of the model
(the base $0$, introduced in passing from $\cR_j$ to $\cR_j'$ above,
does not change the fact that
each coefficient is a $(B,\nu)$-bounded function,
since $\nu\ge \mu_1^{1/2}(B)$ and hence $\nu\ge 0$).
[Again, we use the fact that one can assume $\mu_1(B)\ge 1$.]

Furthermore the coefficients $c_i$ of this asymptotic expansion
vanish if $i<\ord(X)$, since those of each $F_1(\mec K|_{X_1},n)$
do.
\end{proof}

\subsection{Proof of Theorem~\ref{th_main_certified_pairs}}

Before giving the proof, we need the following lemma.

\begin{lemma}\label{le_vlg_compare}
Let $T$ be a graph, and $\mec k,\mec k'$ be
two maps $E_T\to\naturals$ with $\mec k\le \mec k'$
(i.e., $k(e)\le k'(e)$ for all $e\in E_T$).  Then
\begin{equation}\label{eq_vlg_compare}
\mu_1\bigl( \VLG(T,\mec k) \bigr)
\ge
\mu_1\bigl( \VLG(T,\mec k') \bigr) \ .
\end{equation}
\end{lemma}
Its proof
is a standard consequence of ``Shannon's algorithm,''
and {\em majorization}
as described just above Theorem~3.5 of
\cite{friedman_alon}.
In the terminology there,
each entry of the matrix $Z_G(z)$, where $G$ is the oriented line graph of
$\VLG(T,\mec k)$, majorizes each of $Z_H(z)$ where $H$ is the oriented
line graph of
$\VLG(T,\mec k')$; hence each entry of $M_G(z)$ majorizes that of
$M_H(z)$; hence
equation~(12) and Theorem~3.5 of \cite{friedman_alon}
imply \eqref{eq_vlg_compare}.

\begin{proof}[Proof of Theorem~\ref{th_main_certified_pairs}]
It suffices to prove Theorem~\ref{th_main_certified_pairs} when we 
restrict to counting pairs of a given pair homotopy type
$X^{\rm pairH}$, since all pairs counted in
\eqref{eq_subgraphs_times_walks} are of a finite number of 
possible pair homotopy types.

So fix a pair homotopy type
$X^{\rm pairH}=(X^\le,{X_1}^\le,\psi_\Bg^\og)$ (there is no harm in assuming
the last element of $X^{\rm pairH}$ equals $\psi_\Bg^\og$, since it
is isomorphic to $\psi_\Bg^\og$).
If $(w,\tilde S_\Bg^\og)$ is a pair of this pair-homotopy,
then the homotopy type of the walk $T^\le$ is determined by ${X_1}^\le$ and,
according to \eqref{eq_edge_lengths_relation}
the condition $\mec k\ge \bec\xi$ is equivalent to
\begin{equation}\label{eq_pair_inequality_sum}
\xi(e_T) \le \sum_{e_B\in e_T} K(e_B),
\end{equation}
using the notion $e_B\in e_T$ in \eqref{eq_edge_lengths_relation}.

Next, with the same $e_B\in e_T$ notation,
consider all vectors $\Xi\from E_X\to\naturals$ such that for each
$e_T\in E_T$ we have
\begin{equation}\label{eq_pair_equality_sum}
\xi(e_T) = \sum_{e_B\in e_T} \Xi(e_B);
\end{equation}
since each $e_T$ contains a fixed subset of $e_B$ in $E_B$,
for each $e_T$ there are only finitely many possible values of $\Xi(e_B)$
satisfying \eqref{eq_pair_equality_sum}; since each $e_B$ lies in
some $e_T$, there are only finitely many such vectors $\Xi$;
denote these vectors by
$\Xi_1,\ldots,\Xi_s$.

Any $\mec K$ satisfying \eqref{eq_pair_inequality_sum} must satisfy
$\mec K\ge \bec\Xi_i$ for at least one $i$.  For each subset
$S\subset [s]$, let 
$$
\Xi_S \eqdef \max_{s\in S} \Xi_s,
$$
where the maximum is taken component by component.  In this way
$\mec K\ge \Xi_s$ for all $s\in S$ is equivalent to the condition
$\mec K\ge \Xi_S$.
For any fixed $n$, inclusion-exclusion implies that
$$
f(k,n)=\EE_{G\in\cC_n(B)}
[{\rm snbc}(T^\og,\ge\bec\xi;\psi_B^\og;G,k)]
$$
equals
\begin{equation}\label{eq_inclusion_exclusion_Xis}
\sum_{S\subset [s],\ S\ne\emptyset} (-1)^{1+(\#S)} f_S(k,n), 
\quad\mbox{where}\quad
f_S(k,n)\eqdef \EE_{G\in\cC_n(B)}
[{\rm snbc}(X^{\rm pairH},\ge\bec\Xi_S;G,k)] \ .
\end{equation} 
By Lemma~\ref{le_pair_certified_thm}, setting
$$
\nu_S = \max\Bigl( \mu_1^{1/2}(B), 
\mu_1\bigl(\VLG({X_1},{\bec\Xi_S}|_{X_1})\bigr) \Bigr).
$$
each of the functions
$f_S(k,n)$ has $(B,\nu_S)$-bounded expansions to any order, whose bases
of coefficients are those of the model and whose $i$-th order coefficient
vanishes if $i<\ord(X)$.

Since each $S$ in \eqref{eq_inclusion_exclusion_Xis} is nonempty,
each such $S$ has $j\in S$ for some $j\in[s]$; for such $j,S$ we have
$\bec\Xi_S\ge\bec\Xi_j$, and therefore
$$
\mu_1\bigl(\VLG({X_1},{\bec\Xi_S}|_{X_1})\bigr) 
\le
\mu_1\bigl(\VLG({X_1},{\bec\Xi_j}|_{X_1})\bigr) .
$$
But in view of \eqref{eq_pair_equality_sum} we have
$$
\VLG({X_1},{\bec\Xi_j}|_{X_1})
\isom
\VLG(T,\xi).
$$
Hence $\nu_S\le \nu$ with $\nu$ as in 
\eqref{eq_nu_as_in_main_cert_pairs}.
Hence the expansion coefficients of each $f_S(k,n)$ are also
$(B,\nu)$-bounded,
and then \eqref{eq_inclusion_exclusion_Xis}
also has such an expansion to any order whose coefficients are
$(B,\nu)$-bounded.
This proves the existence of expansions of
\eqref{eq_subgraphs_times_walks}
claimed in Theorem~\ref{th_main_certified_pairs},
when \eqref{eq_subgraphs_times_walks} is summed over pairs
of a given pair homotopy type.
Summing over the finite number of pair homotopy types
of walk-subgraph pairs
$X^{\rm pairH}=(X;X_1^\og,{X_2}^\og_\Bg)$ with walks of
homotopy type $T^\og$ and ${X_2}_\Bg^\og=\psi_\Bg^\og$
yields the desired expansion theorem 
for \eqref{eq_subgraphs_times_walks}.

Furthermore, we know that the coefficients $c_i$ of this expansion
vanish whenever $i<\ord(X)$ for all pair homotopy types
$X^{\rm pairH}=(X;X_1^\og,{X_2}^\og_\Bg)$ of a walk-subgraph
pair whose union occurs
in $\cC_n(B)$ (i.e., whose union occurs with nonzero probability
in $\cC_n(B)$ for some $n$, and hence for all $n$ sufficiently large); 
but the smallest such $\ord(X)$ 
is the smallest order of a subgraph, $U_\Bg$, that occurs in $\cC_n(B)$
and contains a subgraph of homotopy type
$T^\og$ and a $B$-graph isomorphic to $\psi_\Bg$.
This establishes the last claim
in Theorem~\ref{th_main_certified_pairs}.
\end{proof}


\providecommand{\bysame}{\leavevmode\hbox to3em{\hrulefill}\thinspace}
\providecommand{\MR}{\relax\ifhmode\unskip\space\fi MR }
\providecommand{\MRhref}[2]{%
  \href{http://www.ams.org/mathscinet-getitem?mr=#1}{#2}
}
\providecommand{\href}[2]{#2}

\end{document}